\documentclass[10pt, conference, compsocconf]{IEEEtran}
\IEEEoverridecommandlockouts
\usepackage{cite}
\usepackage{amsmath,amssymb,amsfonts}
\usepackage{graphicx}
\usepackage{textcomp}
\usepackage{xcolor}
\usepackage{fancyhdr}
\usepackage{soul}
\usepackage{multirow}
\usepackage{diagbox} 
\usepackage{booktabs} 
\newcommand{\tabincell}[2]{\begin{tabular}{@{}#1@{}}#2\end{tabular}}  
\usepackage[noend]{algpseudocode}  
\usepackage{algorithm}
\usepackage{amsthm}

\newtheorem{lemma}{Lemma}
\newtheorem{remark}{Remark}

\newtheorem{example}{Example}

\def\BibTeX{{\rm B\kern-.05em{\sc i\kern-.025em b}\kern-.08em
    T\kern-.1667em\lower.7ex\hbox{E}\kern-.125emX}}

\ifCLASSINFOpdf
\else
\fi
\begin{document}
%
\title{Selfish Mining in Ethereum}

\author{\IEEEauthorblockN{Jianyu Niu and Chen Feng}
\IEEEauthorblockA{School of Engineering, University of British Columbia (Okanagan Campus), Kelowna, Canada\\
\{jianyu.niu, chen.feng\}@ubc.ca}
}

\maketitle

\begin{abstract}
As the second largest cryptocurrency by market capitalization and today's biggest decentralized platform that runs smart contracts, Ethereum has received much attention from both industry and academia. Nevertheless, there exist very few studies about the security of its mining strategies, especially from the selfish mining perspective. 
In this paper, we aim to fill this research gap by analyzing selfish mining in Ethereum and 
understanding its potential threat.
First, we introduce a 2-dimensional Markov process to model the behavior of a selfish mining
strategy inspired by a Bitcoin mining strategy proposed by Eyal and Sirer.
Second, we derive the stationary distribution
of our Markov model and compute long-term average mining rewards.
This allows us to determine the threshold of computational power that makes selfish
mining profitable in Ethereum. We find that this threshold is lower than that in Bitcoin mining
(which is $25\%$ as discovered by Eyal and Sirer), suggesting that Ethereum is more vulnerable to selfish mining than Bitcoin.
\end{abstract}

\IEEEpeerreviewmaketitle

\section{Introduction}

\subsection{Motivation}
The Proof-of-Work (PoW) is the most widely adopted consensus algorithm in blockchain platforms such as Bitcoin \cite{nakamoto2012bitcoin} and Ethereum \cite{buterin2014next}.
By successfully solving math puzzles involving one-way hash functions, the winners of this PoW competition are allowed to generate new blocks that contain as many outstanding transactions as possible (up to the block size limit). As a return, each winner can collect all
the transaction fees and earn a block reward (if its new block is accepted by other participants).
This economic incentive encourages participants to contribute their computation power as much
as possible in solving PoW puzzles---a process often called mining in the literature.

If all the miners follow the mining protocol, each miner will receive block rewards
proportional to its computational power \cite{nakamoto2012bitcoin}. Interestingly, if a set of
colluding miners deviate from the protocol to maximize their own profit, they may obtain a revenue
larger than their fair share. Such behavior is called \emph{selfish mining}  
and has been studied in a seminal paper by Eyal and Sirer in the context of Bitcoin mining \cite{eyal2014majority}. The selfish mining poses a serious threat to any blockchain platform adopting PoW. If colluding miners occupy a majority of the computational power in the system,
they can launch a so-called $51\%$ attack to control the entire system.

Selfish mining in Bitcoin has been well studied with various mining strategies proposed 
(e.g., \cite{sapirshtein2016optimal,nayak2016stubborn,gervais2016security})
and numerous defenses mechanisms suggested (e.g., \cite{zhang2017publish,heilman2014one}).
In sharp contrast, selfish mining in Ethereum 
has not received much attention.
Ethereum differs from Bitcoin in that it provides the so-called uncle and nephew rewards in addition to the (standard) block rewards used in Bitcoin \cite{wood2014ethereum}. 
This complicates the analysis. As a result, most existing research results on Bitcoin cannot be directly applied to Ethereum.

\subsection{Objective and Contributions}
In this paper, we aim to fill this research gap by analyzing selfish mining in Ethereum
and understanding its potential threat. To achieve this goal, we first introduce a 2-dimensional Markov process to model the behavior of a selfish mining
strategy inspired by \cite{eyal2014majority}. 
We then derive the stationary distribution of our Markov model and compute long-term average mining rewards.
This allows us to determine the threshold of computational power which makes selfish
mining profitable in Ethereum. We find that this threshold is lower than that in Bitcoin.
In other words, selfish mining poses a more serious threat to Ethereum due to the presence of uncle and nephew rewards. We finally perform extensive simulations to verify our mathematical results and obtain several engineering insights.

Although our mining strategy is similar to that proposed by Eyal and Sirer [3], our analysis is different from theirs in two aspects. First, our Markov model is 2-dimensional whereas their model is 1-dimensional. Second, our analysis tracks block rewards in a probabilistic way whereas their analysis tracks rewards in a deterministic way. It turns out that our 2-dimensional model, combined with the probabilistic tracking, enables us to characterize the effect of uncle and nephew rewards, which is impossible with the 1-dimensional model and deterministic tracking. 
The main contributions of this paper are summarized as follows:
\begin{itemize}
\item We develop a mathematical analysis that can capture the impact of uncle and nephew rewards with selfish mining in detail.
We believe that our analysis can be extended to study more advanced 
selfish mining strategies.

\item Using our theoretical results, we evaluate the threshold of making selfish mining profitable
under different versions of Ethereum proposals with a particular focus on \textsc{EIP100} (which is adopted by the released Byzantium \cite{wood2014ethereum}).
We find that the threshold is lower than that in Bitcoin, suggesting that Ethereum is 
more vulnerable to selfish mining than Bitcoin.

\end{itemize}

The rest of the paper is organized as follows.
Section~\ref{sec:primer} provides some necessary background for our work.
Section~\ref{sec:background} introduces a selfish mining strategy for Ethereum inspired by Eyal and Sirer.
Section~\ref{sec:analysis} gives the mathematical analysis and results. 
Section~\ref{sec:evaluation} presents our simulation results.
Section~\ref{sec:solution} discusses how to improve the security of Ethereum.
Related work are discussed in Section~\ref{sec:related}.
Finally, Section~\ref{sec:conclude} concludes the paper.

\section{A Primer on Ethereum}\label{sec:primer}
Ethereum is a distributed blockchain-based platform that runs smart contracts. Roughly speaking, a smart contract is a set of functions defined in a Turing-complete environment. The users of Ethereum are called clients. A client can issue transactions to create new contracts, to send Ether  (internal cryptocurrency of Ethereum) to contracts or to other clients, or to invoke some functions of a contract. The valid transactions are collected into blocks; blocks are chained together through each one containing a cryptographic hash value of the previous block. 

There is no centralized party in Ethereum to authenticate the blocks and to execute the smart contracts. Instead, a subset of clients (called miners in the literature) verify the transactions, generate new blocks, and use the PoW algorithm to reach consensus, receiving Ethers for their effort in maintaining the network. 

\subsection{Blockchain}
Each block in the Ethereum blockchain contains three components: a block header, a set of transactions, and some reference links to certain previous blocks called uncle blocks (whose role will be explained in Sec. \ref{sec:mining rewards}) \cite{wood2014ethereum}. The block header includes a Keccak 256-bit hash value of the previous block, a timestamp, and a nonce (whose role will be explained shortly). See Fig.~\ref{fig:blockchain} for an illustration in which blocks are linked together by the hash references, forming a chain structure. 

Such a chain structure has several desirable features. First, it is tamper-free. Any changes of a block will lead to subsequent changes of all later blocks in the chain. 
Second, it prevents double-spending. All the clients will eventually have the same copy of the blockchain\footnote{More precisely, all the clients will have a ``common prefix'' of the blockchain \cite{garay2015bitcoin}.} so that any transactions involving double-spending will be detected and discarded. 

\begin{figure}[!ht]
\setlength{\abovecaptionskip}{2pt}
\setlength{\belowcaptionskip}{5pt}
\centering
\includegraphics[width=0.9\linewidth]{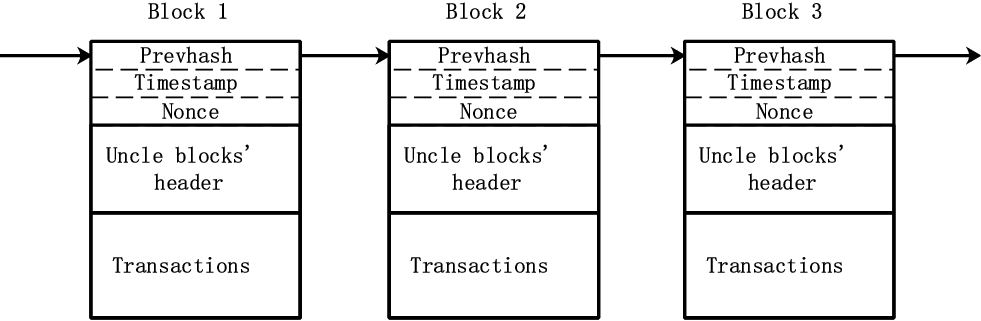}
\caption{An illustration of the blockchain structure in Ethereum.}
\vspace{-2mm}
\label{fig:blockchain}
\end{figure}

\subsection{PoW and Mining} \label{sec:b_mining}
In order for a miner to produce a new valid block in PoW, it needs to find a value of the nonce such that the hash value of the new block is below a certain threshold depending on the difficulty level---a system parameter that can be adjusted. 
This puzzle-solving process is often referred to as \emph{mining}.  
Intuitively, the mining difficulty determines the chance of finding a new block on each try. By adjusting the mining difficulty, the blockchain system can maintain stable chain growth. 

Once a new block is produced, it will be broadcast to the entire network. In the ideal case, a block will arrive at all clients before the next block is produced. If this happens to every block, then each client in the system will have the same chain of blocks. In reality, the above ideal case doesn't always happen. For example, if a miner produces a new block \emph{before} he or she receives the previous block, a fork will occur where two ``child'' blocks share a common ``parent" block. See Fig.~\ref{fig:fork} for an illustration. 
In general, each client in the system observes a tree of blocks due to the forking. As a result, each client has to choose a main chain from the tree according to certain rules (e.g., the longest chain rule in Bitcoin and the heaviest subtree rule in the GHOST protocol)\footnote{Although Ethereum claimed to apply the heaviest subtree rule \cite{sompolinsky2015secure}, it seems to apply the longest chain rule instead \cite{gervais2016security}.}. The common prefix of all the main chains is called the \emph{system main chain}---a key concept that will be used in our analysis.

\begin{figure}[t]
\centering
\includegraphics[width=0.9\linewidth]{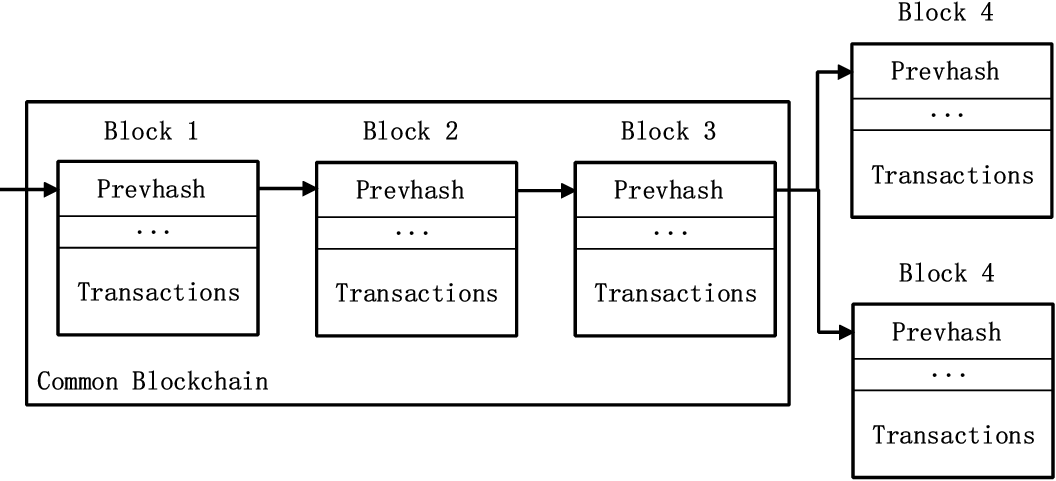}
\caption{An illustration of a forked blockchain.}
\vspace{-2mm}
\label{fig:fork}
\end{figure}

\subsection{Ethereum Milestones} \label{sec:milestone}
Ethereum has four milestones: Frontier, Homestead, Metropolis, and Serenity. We are now in the third milestone where the mining difficulty level depends not only on the growth of the system main chain but also on the appearance of uncle blocks, as suggested in \textsc{EIP100} which is adopted by the released Byzantium \cite{wood2014ethereum}. By contrast, the mining difficulty level of Bitcoin only depends on the growth of the system main chain. Such a difference motivates our work. 

\section{Selfish Mining on Ethereum} \label{sec:background}
\subsection{Mining Model} \label{sec:mining}
In this paper, we consider a system of $n$ miners. The $i$th miner has $m_{i}$ fraction of total hash power. Clearly, we have $\sum_{i=1}^{n}m_{i} = 1$. We assume miners are either honest (those who follow the protocol) or selfish (those who deviate from the protocol in order to maximize their own profit). Let $\mathcal{S}$ denote the set of selfish miners and $\mathcal{H}$ denote the set of honest miners. Let $\alpha$ denote the fraction of total hash power controlled by selfish miners and $\beta$ denote the fraction of total hash power controlled by honest miners. We have $\alpha = \sum_{i \in \mathcal{S}} m_i$ and $\beta = \sum_{i \in \mathcal{H}} m_i$. Clearly, $\alpha + \beta = 1$. Without loss of generality, we assume a single selfish mining pool with $\alpha$ fraction of hash power. 

The PoW mining process can be viewed as a series of Bernoulli trials, each of which independently finds a valid nonce to generate a new block with the same probability, depending on the difficulty level explained in Sec.~\ref{sec:b_mining}. Recall that a Bernoulli process can be approximated by a Poisson process if the duration of a trial is very short and the success probability is very low \cite{gallager2013stochastic}. Both conditions are held in the context of Bitcoin and Ethereum. 
Therefore, we can model the mining process of the $i$th miner as a Poisson process with rate $f m_{i}$. Here, $f$ denotes the block mining rate of the entire system (e.g., $10$ minutes per block in Bitcoin and $10-20$ seconds per block in Ethereum). That is, the $i$th miner generates new blocks at rate $f m_{i}$. Hence, the selfish pool generates blocks at rate $f \alpha$, and the honest miners generate blocks at rate $f \beta$.  This model
has been widely used in the literature. See, e.g.,  \cite{nakamoto2012bitcoin,Papadis2018info, bagaria2018deconstructing}.

\subsection{Mining Rewards} \label{sec:mining rewards}
There are three types of block rewards in Ethereum, namely, static block reward, uncle block reward, and nephew block reward \cite{buterin2014next, ethereum_reward}, as outlined in Table~\ref{tab:rewards}. The static reward is used in both Ethereum
and Bitcoin. To explain static reward, we introduce the concepts of \emph{regular} and \emph{stale} blocks.
A block is called regular if it is included in the system main chain, and is called stale block otherwise. Each regular block in Ethereum can bring its miner a reward of exactly 3.0 Ethers as an economic incentive. 

The uncle and nephew rewards are unique in Ethereum. An uncle block is a stale block that is a ``direct child'' of the system main chain.
In other words, the parent of an uncle block is always a regular block. An uncle block receives a certain reward if it is referenced
by some future regular block, called a nephew block, through the use of reference links. See Fig.~\ref{fig:ghost} for an illustration
of uncle and nephew blocks. The values of uncle rewards depend on the ``distance'' between the uncle and nephew blocks.
This distance is well defined because all the blocks form a tree. For instance, in Fig.~\ref{fig:ghost}, 
the distance between uncle block $B3$ (uncle block $D2$, resp.) and its nephew block is $1$ ($2$, resp.).
In Ethereum, if the distance is $1$, the uncle reward is $\frac{7}{8}$ of the (static) block reward; if the distance is $2$,
the uncle reward is $\frac{6}{8}$ of the block reward and so on. Once the distance is greater than $6$, the uncle reward will be zero. By contrast, the nephew reward is always $\frac{1}{32}$ of the block reward. In addition to blocks rewards, miners can also receive gas costs as a reward for verifying and executing all the transactions\cite{buterin2014next}. However, the gas cost is dwarf with other rewards, and so we ignore it in our analysis. 

We use $K_s$, $K_u$, and $K_n$ to denote static, uncle, and nephew rewards, respectively. Without loss of generality, we assume that $K_s = 1$ so that $K_u$ ($K_n$, resp.) represents the ratio of uncle  reward (nephew reward, resp.) to the static reward. As we explained before, in the current version of Ethereum, $K_n < K_u < 1$, and $K_u$ is a function of the distance. As we will see later, our analysis allows $K_u$ and $K_n$ to be any functions of the distance.

\begin{table}
\caption{Mining Rewards in Ethereum and Bitcoin}
\label{tab:rewards}
\begin{center}
\begin{tabular}{cccl}
\toprule
    & Ethereum   & Bitcoin & Purpose \\
\midrule
    \tabincell{c}{Static Reward}      & $\checkmark$ & $\checkmark$ &   \tabincell{c}{Compensate for miners'\\mining cost}\\ 
    Uncle Reward    & $\checkmark$ & $\times$ & \tabincell{c}{Reduce centralization \\ trend of mining}\\  
    Nephew Reward   & $\checkmark$ & $\times$ & \tabincell{c}{Encourage miners to  \\ reference uncle blocks}\\
        \tabincell{c}{Transaction Fee \\ (Gas Cost)}& $\checkmark$ & $\checkmark$ & \tabincell{c}{Transaction execution; \\Resist network attack}\\ 
  \bottomrule
\end{tabular}
\end{center}
\end{table}

\begin{figure}[t]
\centering
\includegraphics[width=0.9\linewidth]{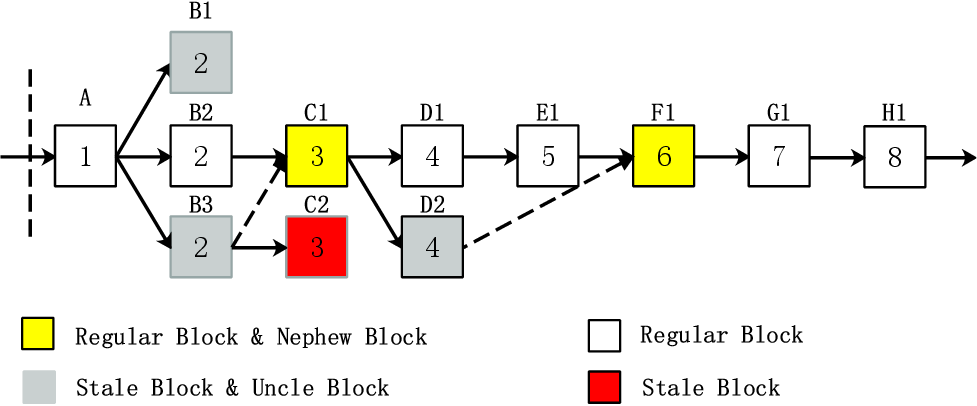}
\caption{Different block types in Ethereum. Here, regular blocks include $\{A, B2, C1, D1, E1, F1, G1, H1\}$ and stale blocks include $\{B1, B3, C2, D2\}$.  Similarly, uncle blocks are $\{B1, B3, D2\}$
and nephew blocks are $\{C1, F1\}$. Uncle block B3 (uncle block D2, resp.) is referenced with distance one (two, resp.).}
\vspace{-2mm}
\label{fig:ghost}
\end{figure}

\subsection{Mining Strategy} \label{sec:strategy}
We now describe the mining strategies for honest and selfish miners. 
The honest miners follow the protocol given in Sec.~\ref{sec:b_mining}.
Each honest miner observes a tree of blocks. It chooses a main chain from
the tree and mines new blocks on its main chain. Once a new block is produced,
the miner broadcasts the block to everyone in the system. Also, it includes as
many reference links as possible to (unreferenced) uncle blocks in the tree.

\begin{figure}[t]
\centering
\includegraphics[width=0.75\linewidth]{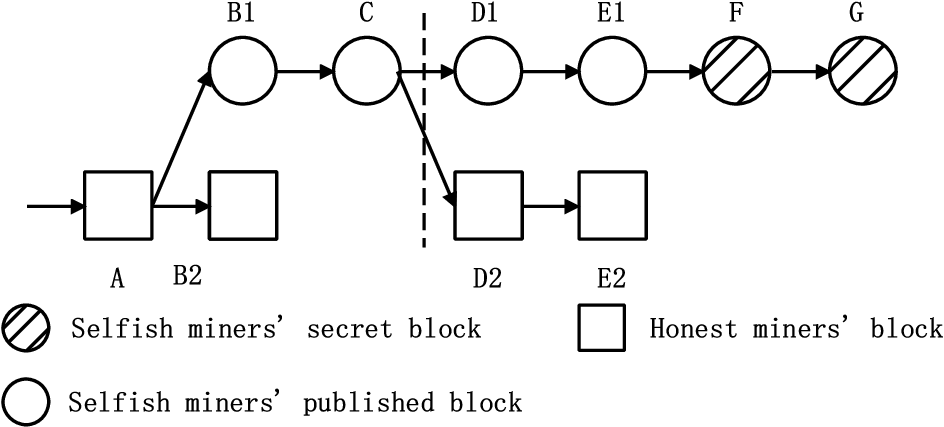}
\caption{An example to illustrate private branch length and public branch length.}
\vspace{-2mm}
\label{fig:lead}
\end{figure}

By contrast, the selfish pool can withhold its newly mined blocks and publish them strategically to maximize its own revenue.
The basic idea behind selfish mining is to increase the selfish pool's share of static rewards 
and, at the same time, to gain as many uncle and nephew rewards as possible. 
Specifically, the selfish pool keeps its newly discovered blocks private, creating a fork on purpose. The pool then
continues to mine on this private branch, while honest miners still mine on public branches (which are often shorter than the private branch). 

Fig.~\ref{fig:lead} gives an example in which ``circle'' blocks are mined by the pool and ``square'' blocks are mined by honest miners. In this example, the private branch consists of $4$ blocks $(D1, E1, F, G)$, all of which are mined by the pool with $(D1, E1)$ published and $(F, G)$ still private. There are two public branches, namely, $(D1, E1)$ and
$(D2, E2)$, because honest miners can see both branches. Each honest miner then chooses one public branch to mine new blocks
according to certain rules (e.g., the longest chain rule). Here, the two public branches are of equal length. This is not a coincidence. In fact, we can show that public branches always have
the same length under our selfish mining strategy.

Let $L_s(t)$ be the length of the private branch seen by the selfish pool at time $t$.
Similarly, let $L_h(t)$ be the length of public branches seen by honest miners at time $t$. 
(Note that $L_h(t)$ is well defined because all public branches have the same length.) We are now ready to describe our selfish mining strategy
which is based on the strategy in \cite{eyal2014majority}.

\begin{algorithm}[ht]
		\caption{An selfish Mining Strategy in Ethereum}
		\label{Algo1}
		\begin{algorithmic}[1]
			
			\Statex \hspace{-1.8em}\textbf{on}  The selfish pool mines a new block 
			\State reference all (unreferenced) uncle blocks based on its private branch
			\State $L_s \leftarrow L_s + 1$
			    \If{$(L_s, L_h) = (2, 1)$}
			        \State  publish its private branch
			        \State  $(L_s, L_h) \leftarrow (0, 0)$ (since all the miners achieve a consensus)
			    \Else
			        \State  keep mining on its private branch
			    \EndIf
			\Statex
			
			\Statex \hspace{-1.8em}\textbf{on}  Some honest miners mine a new block 
			\State The miner references all (unreferenced) uncle blocks based on its public branches
			\State  $L_h \leftarrow L_h + 1$    
			\If{$L_s < L_h$}
			    \State  $(L_s, L_h) \leftarrow (0, 0)$
			    \State keep mining on this new block
			\ElsIf{$L_s = L_h$ }
			    \State   publish the last block of the private branch
			\ElsIf{$L_s = L_h + 1$ }
			    \State   publish its private branch
			    \State  $(L_s, L_h) \leftarrow (0, 0)$ (since all the miners achieve a consensus)
			\Else
			    \State   publish first unpublished block in its private branch
			    \State   set $(L_s, L_h) = (L_s - L_h + 1, 1)$ if the new block is mined on 
			    a public branch that is a prefix of the private branch
			\EndIf
			
		\end{algorithmic}
\end{algorithm}

Algorithm~\ref{Algo1} presents the mining strategy. When the selfish pool mines a new block (see lines $1$ to $7$),
it will keep this block private and continue mining on its private branch until its advantage is 
very limited (i.e., $(L_s, L_h) = (2, 1)$) which will be discussed later.

When some honest miners mine a new block, the length of a public branch will be increased by $1$.
We have the following cases. Case 1) If the new public branch is longer than the private branch, 
the pool will adopt the public branch and mine on it. (That is why the pool will set $(L_s, L_h) = (0, 0)$.)
Case 2) If the new public branch has the same length as the private branch, the pool will publish
its private block immediately hoping that as many honest miners will choose its private branch as possible
(since honest miners will see two branches of the same length when the private branch is published).
Case 3) If the new public branch is shorter than the private branch by just $1$, the pool will
publish its private branch so that all the honest miners will adopt the private branch.
Case 4) If the new public branch is shorter than the private branch by at least $2$, the pool 
will publish the first unpublished block since the pool still has a clear advantage.
Moreover, if the new block is mined on a public branch that is a prefix of the private branch,
the pool will set $(L_s, L_h) = (L_s - L_h + 1, 1)$ due to a new forking (caused by the honest miner).

To better illustrate the selfish mining strategy, we provide an example in Fig.~\ref{fig:mining}.
In Step 1, we have $(L_s, L_h) = (3, 0)$. In Step 2, some honest miner publishes block $A2$
and we have $(L_s, L_h) = (3, 1)$. This corresponds to Case 4). Hence, the pool immediately publishes
block $A1$, still having an advantage of $2$ blocks. In Step 3, some honest miner publishes block $B2$,
leading to $(L_s, L_h) = (3, 2)$. This corresponds to Case 3). Thus, the pool publishes its private branch, making honest miners' blocks ($A2$ and $B2$) stale.

\begin{remark}
The selfish mining strategy presented above isn't necessarily optimal. By studying its behavior, we hope to reveal some characteristics of the selfish mining in Ethereum. 
\end{remark}

\begin{figure*}[tbh]
    \center
    \begin{tabular}{ccc}
        \includegraphics[width=0.2\textwidth]{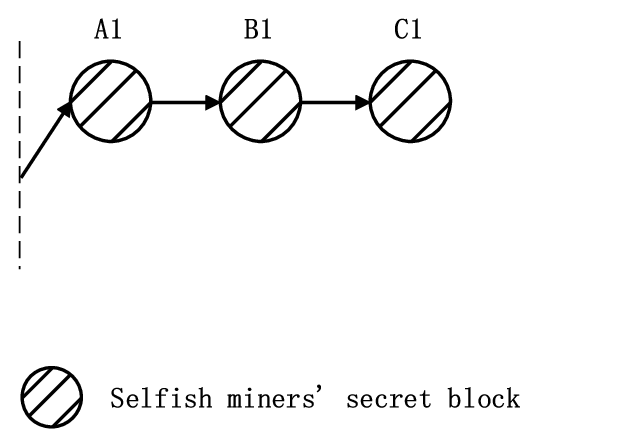} &
        \includegraphics[width=0.2\textwidth]{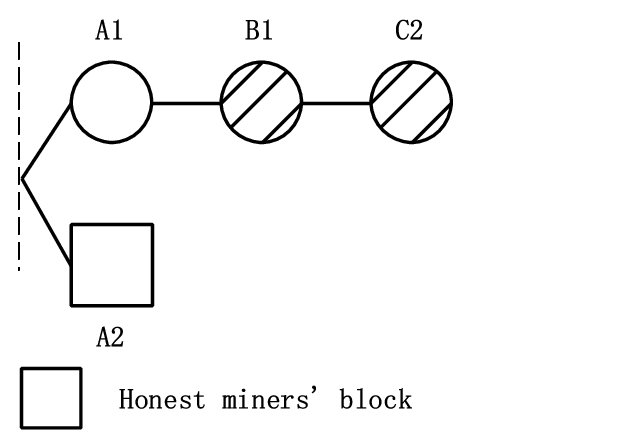} & 
        \includegraphics[width=0.2\textwidth]{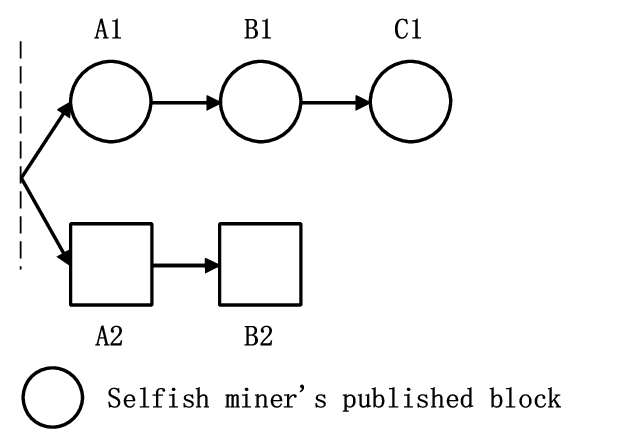} \\
            a) Step 1: selfish pool withholds 3 blocks & b) Step 2: selfish pool publishes 1 block & b) Step 3: selfish pool overrides 2 blocks
    \end{tabular}
    \caption{A simple example of the mining strategy.}
    \label{fig:mining}
\end{figure*}

\subsection{Mining Pool} \label{sec:pool}


\begin{figure}[ht]
\centering
\includegraphics[width=0.8\linewidth]{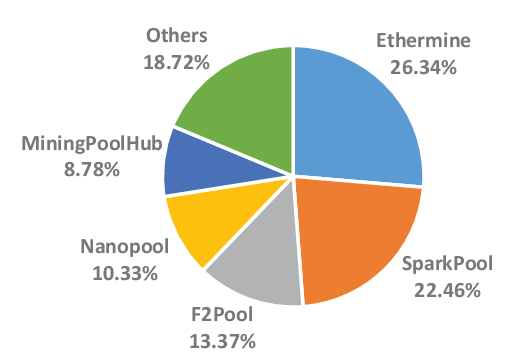}
\caption{The top 5 mining pools' hash power in Ethereum (2018.09).}
\vspace{-2mm}
\label{fig:pool}
\end{figure}

In Ethereum, individual miners can form mining pools to mine blocks together
and share the revenue according to individuals' hash power. Fig.~\ref{fig:pool} 
presents the fractions of the hash power of various mining pools in Ethereum \cite{ethereumscan}. 
The largest mining pool (called Ethermine) has dominated $26.34\%$ of the total hash power.
The top two mining pools have dominated $48.8\%$ of the total hash power.
The top five mining pools have more than $81\%$ of the total hash power. 
Although these mining pools are not necessarily selfish, their presence motivates
us to understand the impact of selfish mining in Ethereum.

\section{Analysis of Selfish Mining} \label{sec:analysis}
In this section, we will study the long-term behavior of the selfish mining strategy using a Markov model
with a particular focus on the mining revenue. 

\begin{figure*}[t]
\centering
\includegraphics[width=0.8\linewidth]{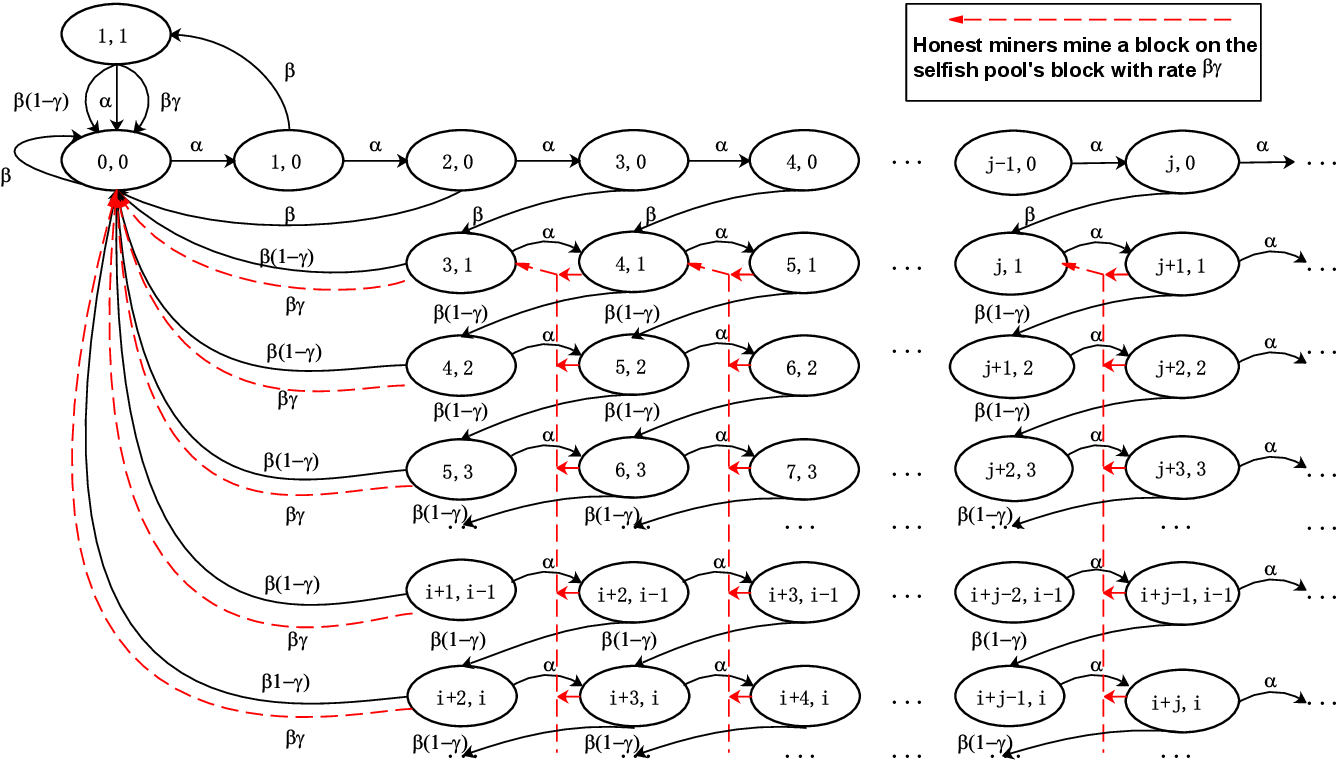}
\caption{The Markov process of the selfish mining in Ethereum.}
\vspace{-2mm}
\label{fig:state}
\end{figure*}

\subsection{Network Model}

To simplify our analysis, we follow the network model of \cite{eyal2014majority, sapirshtein2016optimal, gervais2016security}
which assumes that the time it takes 
to broadcast a block is negligible. In particular, we introduce the same parameter $\gamma$
as in \cite{eyal2014majority}, which denotes the ratio of honest miners that 
are mining on blocks produced by the selfish pool (rather than by the honest miners) 
whenever they observe 
a fork of two branches of equal length.
For example, if the honest miners apply the uniform tie-breaking rule (when they observe 
a fork of two branches of equal length), then $\gamma = \frac{1}{2}$.
On the other hand, if the pool can launch a network attack to influence honest miners' block propagation,
then only a few honest miners can see any new block produced by some honest miners.
In this case, the parameter $\gamma$ is close to $1$. 
Therefore, the parameter $\gamma$ captures the pool's communication capability.
In this paper, we assume that $\gamma$ takes values in the interval $[0, 1]$.

\subsection{Markov Process}

For ease of presentation, we re-scale the time axis so that the selfish pool generates new blocks at rate $\alpha$
and the honest miners generate new blocks at rate $\beta$. We are now ready to define the system state.
Recall that $L_s(t)$ is the length of the private branch 
and $L_h(t)$ is the length of the public branches at time $t$. 
Clearly, $\left ( L_s(t), L_h(t) \right )$ captures the system state at time $t$.
The state space contains the following states:
$(0, 0)$, $(1, 0)$, $(1, 1)$, as well as $(i, j)$ with $i - j \ge 2$ and $j \ge 0$.
It is easy to verify that $\left ( L_s(t), L_h(t) \right )$ evolves as a Markov process
under our selfish mining strategy and the network model, as illustrated in Fig.~\ref{fig:state}. Moreover, we can show that the process $\left ( L_s(t), L_h(t) \right )$ is positive recurrent
and so it has a unique stationary distribution.

\subsection{The Stationary Distribution} \label{sec:distribution}
To compute the stationary distribution of the process $\left ( L_s(t), L_h(t) \right )$, 
we need to derive the transition rates for the state evolution. The results are provided below.

\begin{itemize}
\item $q_{(0, 0), (0, 0)} = \beta$ \\
This transition happens if any honest miner produces a new block, broadcasts it
to everyone. Then, the selfish pool adopts the public branch and mines on it. 
Thus, the rate is $\beta$.

\item $q_{(0, 0), (1, 0)} = \alpha$ \\
This transition happens if the pool produces a new block and keeps it private. 
Thus, the rate is $\alpha$.

\item $q_{(1, 0), (2, 0)} = \alpha$ \\
This transition happens if the pool produces a new block and keeps it private. 
Thus, the rate is $\alpha$.

\item $q_{(1, 0), (1, 1)} = \beta$ \\
This transition happens if any honest miner produces a new block and the pool
immediately publishes its private block (because the new public branch
has the same length as the private branch). Thus, the rate is $\beta$.

\item $q_{(1, 1), (0, 0)} = \alpha + \beta  = 1$ \\
This transition happens if any of the following events happens.
1) The pool produces a new block and publishes its private branch (because $(L_s, L_h) 
= (2, 1)$);
2) Any honest miner produces a new block and the pool has to publish its private branch
(because $L_s = L_h + 1$). 
Thus, the rate is $\alpha + \beta = 1$.

\item $q_{(i, j), (i+1, j)} = \alpha$ for $i \geq 2$ and $j \geq 0$\\
This transition happens if the pool produces a new block and keeps it
private. Thus, the rate is $\alpha$.

\item $q_{(i, j), (i-j, 1)} = \beta \gamma$ for $i - j \geq 3$ and $j \geq 1$\\
This transition happens if any honest miner mines a new block on a public branch
which is a prefix of the private branch. Then, the pool publishes 
a private block accordingly. Thus, the rate is $\beta \gamma$.

\item $q_{(i, j), (0, 0)} = \beta$ for $i - j = 2$ and $j \geq 1$ \\
This transition happens if any honest miner produces a new block and then the pool
publishes its private branch (because $L_s = L_h + 1$). Thus, the rate is $\beta$.

\item $q_{(2, 0), (0, 0)} = \beta$ \\
This transition happens if any honest miner produces a new block and then the pool
publishes its private branch (because $L_s = L_h + 1$). Thus, the rate is $\beta$.

\item $q_{(i, 0), (i, 1)} = \beta$ for $i \geq 3$\\
This transition happens if any honest miner produces a new block and the pool
publishes a private block. Thus, the rate is $\beta$.

\item $q_{(i, j), (i, j+1)} = \beta(1 - \gamma)$ for $i - j \geq 3$ and $j \geq 1$ \\
This transition happens if any honest miner 
mines a new block on a public branch which is not a prefix of the private branch.
Thus, the rate is $\beta ( 1 - \gamma)$.
\end{itemize}

Let $\{ \pi_{i,j} \}$ be the steady-state distribution of the Markov process $\left ( L_s(t), L_h(t) \right )$. Then, the set
$\{ \pi_{i,j} \}$ satisfies the following global balance equations according to the above transition rates.
\begin{equation}
\label{eq:state}
\left\{
\begin{array}{lr}
\alpha \pi_{0,0} = \pi_{1,1} + \beta \sum_{j = 0}^{\infty} \pi_{2+j,j}, \\
\pi_{1,1} = \beta\pi_{1,0}, \\
\pi_{3,1} = \beta\pi_{3,0} +  \sum_{j = 1}^{\infty} \beta \gamma \pi_{3+j,j}, \\
\pi_{i,0} = \alpha \pi_{i-1,0}, \mbox{for }  i \geq 1, \\
\pi_{i,1} = \beta\pi_{i,0} +  \alpha \pi_{i-1,1} + \sum_{j = 1}^{\infty} \beta \gamma \pi_{i+j,j} , \mbox{for }  i \geq 4,\\
\pi_{i,i-2} = \beta \left( 1 -\gamma \right)  \pi_{i,i-1} , \mbox{for }  i \geq 4,\\
\pi_{i,j} = \alpha \pi_{i-1,j} + \beta \left( 1 -\gamma \right) \pi_{i,j-1} , \mbox{for }  j \geq 2, i \geq 5 . \\
\end{array}
\right.
\end{equation}

Solving the global balance equations, we obtain the following results for the stationary distribution:
\begin{align*}
    \pi_{0, 0} &= \frac{1-2\alpha}{2\alpha^{3} - 4\alpha^{2} + 1}, \\
    \pi_{i,0} &= \alpha^{i} \pi_{0,0} \mbox{ for } i \geq 1, \\
    \pi_{1,1} &= \left( \alpha -\alpha^{2} \right) \pi_{0,0}, \\
    \pi_{i,j} &= \alpha^{i} \left( 1- \alpha \right)^{j}\left( 1- \gamma \right)^{j} f(i,j,j)
    \pi_{0,0} + \\
    &\quad \alpha^{i-j} \gamma \left( 1- \gamma \right)^{j-1} \left( \frac{1}{\left( 1- \alpha \right)^{i-j-1}} - 1\right) \pi_{0,0}  - \\
    &\quad \gamma \left( 1- \gamma \right)^{j-1} \sum_{k=1}^{j}   \alpha^{i-k}\left( 1- \alpha \right)^{j-k}f(i,j,j-k) \pi_{0,0}
\end{align*}
for $i\geq j+2$ and $j\geq 1$, where the function $f(x, y, z)$ is defined as 

\begin{equation}
\label{eq:fcn}
f(x,y,z) =
\left\{
\begin{array}{lr}
\underset{z}{\underbrace{\sum_{s_{z} = y+2}^{x}\sum_{s_{z-1} = y+1}^{s_{z}} .. \sum_{s_{1} = y-z+3}^{s_{2}}}}{1}, \hfill z\geq 1, x\geq y+2,  \\
0,  \hfill \mbox{otherwise}.
\end{array}
\right.
\end{equation}
The function $f(x, y, z)$ is a  multiple summations. See Appendix~\ref{appen:function} for some
concrete examples. 

We have the following remarks for the stationary distribution in Equation~\eqref{eq:fcn}.
\begin{remark}
When $0 < \alpha < \frac{1}{2}$, we have $0 < \pi_{0, 0} < 1$. Specifically, the distribution $\pi_{0, 0}$ only depends on the parameter $\alpha$ and is monotonically decreasing. The tendency of $\pi_{0, 0}$ suggests that with more hash power, the selfish pool can obtain more lead blocks and so stay in state $(0, 0)$ less frequently.
\end{remark}

\begin{remark}
The distributions $\pi_{i, 0}$ with $i \geq 1$ are decreasing geometrically and are less than $10^{-6}$ when $i \geq 15$ when $\alpha = 0.4$. It suggests that we can truncate the states in the numerical calculation. 
\end{remark}


\subsection{Reward Analysis} \label{sec:regular}
In this subsection, we conduct the reward analysis for each state transition. Our analysis differs from the previous analysis (e.g., \cite{eyal2014majority,nayak2016stubborn}) in that we track various block rewards in a probabilistic way. Recall that each state transition induces a new block (mined by an honest miner or the pool). In general, it is impossible to decide the number of rewards associated with this new block when it is just created because the ``destiny" of this new block depends on the evolution
of the system. For this reason, we will instead compute the expected rewards for the new block. 
In contrast, the previous analysis tracks published blocks associated with a state transition (whose destiny is already determined)  
rather than the new block and so it can compute the exact rewards. This gives rise to the following two questions.

\begin{enumerate}
\item What is wrong with tracking published blocks?
\item How shall we compute the expected rewards for a new block at the time of its creation?
\end{enumerate}

To answer the first question, one shall notice that tracking published blocks don't provide enough information to compute the uncle and nephew rewards. Recall from Sec.~\ref{sec:mining rewards} that a published regular block can receive nephew rewards by referencing outstanding uncle blocks. The amount of nephew rewards depends on the number of outstanding uncle blocks. As such, we need to
keep track of all the outstanding uncle blocks in the system together with their depth information (which is needed to determine the number
of uncle rewards). This greatly complicates the state space.  

To answer the second question, one shall notice that it suffices to compute the expected rewards for a new block by using the following
information: the probability that it becomes a regular block, the probability that it becomes an uncle block, the distance to its potential nephew block (if it indeed becomes an uncle block). Perhaps a bit surprisingly, all the information can be determined when this new
block is generated for our selfish mining strategy.

The complete analysis is provided in Appendix~\ref{appen:reward} for a better presentation flow. Here, we just provide a simple example to illustrate our analysis. Assume that the selfish pool has already mined two blocks and kept them private at time $t$. Then, some honest miner generates a new block. According to Algorithm~\ref{Algo1}, the pool publishes its private branch
immediately. As such, this new block will become an uncle block with probability $1$. Furthermore, we can show that this block
will have a distance of $2$ with its potential nephew block. Thus, this new block will receive an uncle reward of $K_u(2)$.
Similarly, its potential nephew block will receive a nephew reward of $K_n(2)$. Moreover, this reward will belong to some
honest miner with probability $\beta (1 + \alpha \beta (1 - \gamma))$ and belong to the pool with probability $1 - \beta (1 + \alpha \beta (1 - \gamma))$.  (See \emph{Case 7} in Appendix~\ref{appen:reward} for details.) Therefore, the expected rewards associated
with this new block are $K_u(2) + K_n(2)$ in total among which $K_u(2) + K_n(2) \beta (1 + \alpha \beta (1 - \gamma))$
rewards will belong to honest miners (and the remaining will belong to the pool).

\subsection{Revenue Analysis} \label{sec:revenue}
In this subsection, we apply the previous reward analysis to compute various rewards received by the selfish pool and honest miners.
This calculation is straightforward. 

\subsubsection{Revenue Computing} \label{sec:revenuecomputing}
First, we compute the static block rewards for the selfish pool (denoted as $r_{b}^{s}$) and honest miners (denoted as $r_{b}^{h}$). We have the following results:
\begin{equation}
\label{incen:in2}
\begin{aligned}
 r_{b}^{s} &= \left( \alpha  (1 - \pi_{0,0}) +  (\alpha^2 + \alpha^2 \beta + \alpha \beta^2 \gamma)  \pi_{0,0} \right) \\
               &= \alpha  - \alpha \beta^2 (1 - \gamma)  \pi_{0,0}, \\
               &= \frac{\alpha (1-\alpha)^2 (4\alpha + \gamma(1 - 2\alpha)) - \alpha^3}{2\alpha^3-4\alpha^2+1}
\end{aligned}      
\end{equation}
and
\begin{equation}
\label{incen:in3}
\begin{aligned}
r_{b}^{h} &= \beta  ( \pi_{0,0} + \pi_{1,1}) + \beta^2 (1 - \gamma)  \pi_{1,0} \\
            &= \frac{(1-2\alpha)(1-\alpha)(\alpha(1-\alpha)(2-\gamma) + 1) }{2\alpha^3-4\alpha^2+1}.
\end{aligned}
\end{equation}
Note that $r_b^s$ and $r_b^h$ represent the long-term average static rewards per time unit.
Since all the miners generate new blocks at rate $1$, the maximum long-term average reward is 
$1$ per time unit. Hence, we have $r_b^s + r_b^h \le 1$.

\begin{remark}
If we only consider static rewards, the above results are the same as those in \cite{eyal2014majority},
though our approach is different from that in \cite{eyal2014majority}.
\end{remark}

Next, we can compute the uncle block rewards for the selfish pool (denoted as $r_{u}^{s}$):
\begin{equation}\label{incen:in4}
\begin{aligned}
r_{u}^{s} &= \alpha \beta^2 (1 - \gamma) K_u(1) \pi_{0,0} \\
            &= \frac{(1-2\alpha)(1-\alpha)^2 \alpha (1-\gamma) }{2\alpha^3-4\alpha^2+1} K_u(1).
\end{aligned}
\end{equation}  

\begin{remark}
Note that $r_{u}^{s}$ is zero in Bitcoin. In other words, selfish pools' blocks without rewards can be viewed as the ``cost" of launching the selfish mining attack. Thus the additional reward in Ethereum will reduce the cost of selfish mining and make it easier. Moreover, as shown in our reward analysis, 
the uncle blocks of the pool are always referenced with distance $1$---the minimum referencing distance possible in the system. Intuitively, this is because the pool has a global view of the system. 
\end{remark}

Similarly, we can compute the uncle block rewards for the honest miners (denoted as $r_{u}^{h}$):
\begin{multline} \label{incen:in5}
    r_{u}^{h} =  (\alpha \beta + \beta^2 \gamma ) K_u(1)  \pi_{1,0} + \sum_{i=2}^{\infty} \beta K_u(i) \pi_{i,0} +  \\
    + \sum_{i = 2}^{\infty} \sum_{j = 1}^{\infty} \beta \gamma K_u(i) \pi_{i+j,j}.
\end{multline}
\begin{remark}
In the current version of Ethereum, the function $K_u(\cdot)$ is given below:
\begin{equation}
\label{eq:fcn1}
K_u(l) =
\left\{
\begin{array}{lr}
(8 - l)/8, & 1 \leq l \leq 6 \\
0,  & \text{otherwise}.
\end{array}
\right.
\end{equation}
Our analysis applies to an arbitrary function of $K_u(\cdot)$.
\end{remark}

Then, we can compute the nephew block rewards for the selfish pool (denoted as $r_{n}^{s}$) and honest miners (denoted as $r_{n}^{h}$):
\begin{equation}\label{incen:in6}
    r_{n}^{s} = \alpha \beta K_s(1) \pi_{1,0} + \sum_{i = 2}^{\infty} \sum_{j = 1}^{\infty} \beta^{i-1} \gamma ( \alpha - \alpha \beta^2 (1 - \gamma)) K_s(i) \pi_{i+j,j}
\end{equation}   
\begin{multline} \label{incen:in7}
    r_{n}^{h} = \alpha \beta^2 (1- \gamma) K_s(1) \pi_{0,0} + \beta^2 \gamma K_s(1) \pi_{1,0} + \\
    +\sum_{i = 2}^{\infty} \sum_{j = 1}^{\infty} \beta^{i} \gamma (1 + \alpha \beta (1 - \gamma)) K_s(i) \pi_{i+j,j}.
\end{multline} 

\begin{remark}
In the current version of Ethereum, the function $K_n(\cdot)$ is always equal to $\frac{1}{32}$.
Our analysis applies to an arbitrary function of $K_n(\cdot)$.
\end{remark}

Finally, we can obtain the total mining revenue $r_{\text{total}}$ as
\begin{equation}\label{incen:in8}
   r_{\text{total}} = r_{b}^{s} +  r_{b}^{h} +  r_{u}^{s} +  r_{u}^{h} +  r_{n}^{s} + r_{n}^{h}.
\end{equation} 
Hence, 
\[
R_{s} = \frac{r_{b}^{s} + r_{u}^{s} + r_{n}^{s}}{r_{\text{total}}}
\]
gives
the share of the mining revenue by the selfish pool.

\subsubsection{Absolute Revenue} \label{sec:absolute}
We now define the absolute revenue $U_s$ for the selfish pool. As we will soon see,
although the absolute revenue is equivalent to the relative revenue (i.e., the share $R_s$)
in Bitcoin, it is different from the relative revenue in Ethereum due to the presence
of uncle and nephew rewards.

Recall that Bitcoin adjusts the mining difficulty level so that the regular 
blocks are generated at a stable rate, say $1$ block per time unit.
Thus, the long-term average total revenue is fixed to be $1$ block reward per time unit with or without selfish mining. This makes the absolution revenue 
equivalent to the relative revenue. 
The situation is different in Ethereum. Even if the regular blocks are generated
at a stable rate, the average total revenue still depends on the generation rate of uncle blocks,
which is affected by selfish mining as we will see shortly.
Indeed, Ethereum didn't take into account the generation rate of uncle blocks 
when adjusting the difficulty level until its third milestone.
This motivates us to consider two scenarios in our analysis: 
1) the regular block generation rate is 
$1$ block per time unit, and 2) the regular and uncle block generation rate is $1$ 
block per time unit.

In our previous analysis, the regular block generation rate is $r_b^s + r_b^h$,
which is smaller than $1$ as explained before. Thus, we can re-scale the time
to make the regular block generation rate to be $1$ block per time unit.
In this scenario, the long-term absolute revenue for the selfish pool is
\begin{equation}\label{incen:in9}
   U_{s} = \frac{r_{b}^{s} + r_{u}^{s} + r_{n}^{s}}{r_b^s + r_b^h},
\end{equation} 
and the long-term absolute revenue for honest miners is
\begin{equation}\label{incen:in10}
   U_{h} = \frac{r_{b}^{h} + r_{u}^{h} + r_{n}^{h}}{r_b^s + r_b^h}. 
\end{equation} 

Similarly, we can re-scale the time to make the regular and uncle block generation rate 
to be $1$ block per time unit and define long-term absolute revenues for the selfish
pool and honest miners accordingly.

\subsubsection{Threshold Analysis}
First of all, if the selfish pool follows the mining protocol, its long-term average absolute revenue
will be $\alpha$, since the network delay is negligible (and so no stale blocks will occur).
On the other hand, if the pool applies the selfish mining strategy proposed in this paper,
its long-term absolute revenue is given by $U_s$, which can be larger than $\alpha$. 

Let $\alpha^*$ be the smallest value such that $U_s \ge \alpha$. That is, $\alpha^*$ is the threshold
of computational power that makes selfish mining profitable in Ethereum. We can determine $\alpha^*$
for both scenarios through numerical calculations. The details will be presented in the next section.

\section{Evaluation} \label{sec:evaluation}
In this section, we build an Ethereum selfish mining simulator to validate our theoretical analysis. In particular, we simulate a system with $n = 1000$ miners, each with the same block generation rate. In our simulations, the selfish pool controls at most $450$ miners (i.e., $\alpha \le 0.45$) and runs our Algorithm~\ref{Algo1}, while the honest miners follow the designed protocol in Sec.~\ref{sec:strategy}. 
Our simulation results are based on an average of $10$ runs, where each run generates $100, 000$ blocks.

\subsection{Validation of the Theory Results}
In this subsection, we validate the long-term average absolute revenues for the selfish pool and honest miners.  Fig.~\ref{fig:Fig1_selfish} plots the results obtained from analysis and simulations. From the results, we can see when $\gamma = 0.5$, $K_u = 4/8K_s$ and $\alpha$ changes from $0$ to $0.45$, the simulation results match our theoretical results\footnote{To simplify the numerical calculations of our results, we only consider the states $(i,j)$ with $i$ and $j$ less than $200$. This approximation turns out to be accurate when $\alpha \le 0.45$.}. In addition, when $\alpha$ is above $0.163$, the selfish pool can always gain higher revenue from selfish mining than following the protocol. More importantly, when $\alpha$ is below the threshold $0.163$, the selfish pool loses just a small amount of revenue due to the additional uncle block rewards, which is quite different from the results in Bitcoin \cite{eyal2014majority}. 

\begin{figure}[ht]
\centering
\begin{minipage}[t]{8cm}
\centering
\setlength{\abovecaptionskip}{-5pt}
\setlength{\belowcaptionskip}{0pt}
\includegraphics[width=3in,height=2.7in]{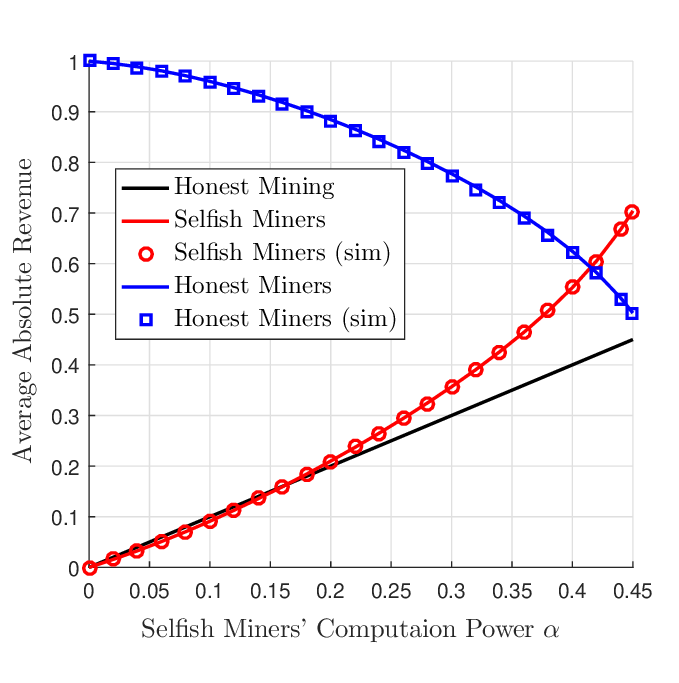}
\caption{Revenue rate for the selfish pool and honest miners when $\gamma = 0.5$, $K_u = 4/8K_s$ and $\alpha$ changes from $0$ to $0.45$.}\label{fig:Fig1_selfish}
\end{minipage}

\hspace{1.5cm}
\begin{minipage}[t]{8cm}
\centering
\setlength{\abovecaptionskip}{-5pt}
\setlength{\belowcaptionskip}{-10pt}
\includegraphics[width=3in,height=2.7in]{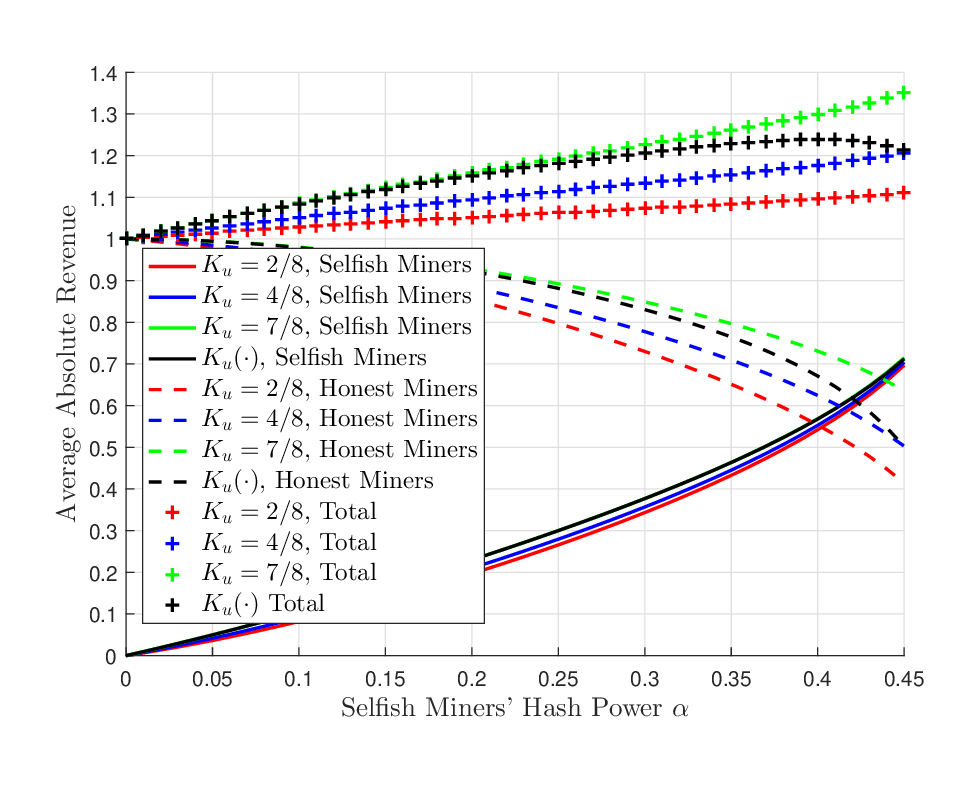}
\caption{Revenue rate for the selfish pool, honest miners under different uncle block reward $K_u$.}\label{fig:Fig2_selfish}
\end{minipage}
\end{figure}

\subsection{Impact of the Uncle Reward}
In this subsection, we explore the impact of the uncle block rewards on the selfish pool's and honest miners' revenues. To this end, we first use the uncle reward function $K_u(\cdot)$ in Ethereum (see Sec.~\ref{sec:revenue} for details) and then set the uncle reward as a fixed value regardless of the distance, ranging from $2/8 K_s$ to $7/8 K_s$. Here, the fixed uncle reward value can directly show its impacts and simplify our understanding. 

Fig.~\ref{fig:Fig2_selfish} shows that the higher the uncle reward, the more absolute revenue for both the selfish pool and honest miners, which is quite intuitive. It also reveals that the total revenue increases with the selfish pool's computation power $\alpha$ and soars to $135\%$ of the revenue without selfish mining, when $K_u = 7/8 K_s$ and $\alpha =0.45$. This is because, without the consideration of uncle blocks into difficulty adjustment, the selfish mining can produce additional uncle and nephew rewards, resulting in the fluctuation of total revenue. Additionally, the uncle reward function $K_u(\cdot)$ used in Ethereum has the same effect as simply setting $K_u = 7/8 K_s$ for selfish pool's revenue (as explained in Sec.~\ref{sec:revenue}).  In contrast, $K_u(\cdot)$ functions complicatedly for the honest miners' revenue. When $\alpha$ is small, its impact is similar to the case of $K_u = 7/8 K_s$, and when $\alpha$ is close to $0.45$, its impact is similar to the case of $K_u = 4/8 K_s$. This is because, with the increase of $\alpha$, the average referencing distances of honest miners' uncle blocks will increase, which further leads to the decrease of honest miners' average uncle rewards when using the function $K_u(\cdot)$. This finding motivates our discussion in Sec.~\ref{sec:solution}.

\begin{figure}[ht]
\centering
\begin{minipage}[t]{8cm}
\centering
\setlength{\abovecaptionskip}{-5pt}
\setlength{\belowcaptionskip}{-10pt}
\includegraphics[width=3in,height=2.7in]{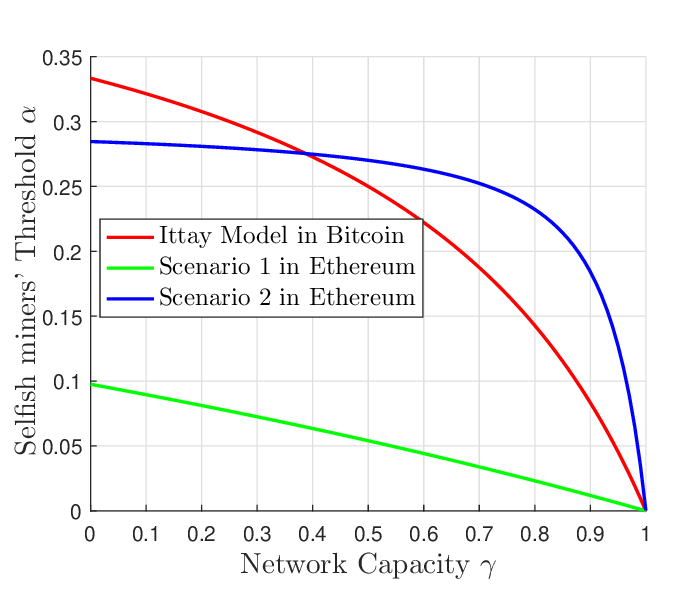}
\caption{The profitable threshold of hash power in Bitcoin and Ethereum.}\label{fig:threshold}
\end{minipage}
\end{figure}

\subsection{Comparison with Bitcoin}
In this subsection, we compare the hash power thresholds of making selfish mining profitable in Ethereum and Bitcoin under different values of $\gamma$. In Ethereum, we use function $K_u(\cdot)$ to compute the uncle rewards. Particularly, we compute the thresholds for the two scenarios described in Sec.~\ref{sec:absolute}: 1) the regular block generation rate is $1$ block per time unit, and 2) the regular and uncle block generation rate is $1$ block per time unit.

From Fig.~\ref{fig:threshold}, we can see that the higher $\gamma$ is, the lower hash power needed for making selfish mining profitable. Specifically, when $\gamma=1$, the selfish mining in Bitcoin and Ethereum can always be profitable regardless of their hash power. Besides that, the results show that the hash power thresholds of Ethereum in scenario $1$ are always lower than in Bitcoin. By contrast, the hash power thresholds in scenario $2$ are higher than Bitcoin when $\gamma \geq 0.39$. This is because the larger $\gamma$ is, the more blocks mined by honest miners are uncle blocks. However, in scenario $2$ the additional referenced uncle blocks will reduce the generation rate of regular blocks, resulting in the decrease of selfish pools' block rewards. Thus the selfish pool needs to have higher hash power in order to make selfish mining profitable. This suggests that Ethereum should consider the uncle blocks into the difficulty adjustment under the mining strategy given in Algorithm~\ref{Algo1}.

\section{Discussion} \label{sec:solution}
In Ethereum, uncle and nephew rewards are initially designed to solve the mining centralization bias---miners form or join in some big mining pools (Sec.~\ref{sec:pool}). This is because, due to propagation delay,  mining pools with huge hash power are less likely to generated stale blocks and can be more profitable for mining. Thus, rewarding the stale block can reduce the mining pools' advantage \cite{uncleincen} and make them less attractive for small miners. However, as analyzed previously, the uncle reward (computed by using the function $K_u(\cdot)$) can greatly reduce the cost of launching selfish mining. To mitigate this issue, here we propose a simple uncle reward function motivated by our analysis in Sec.~\ref{sec:revenuecomputing}. It shows that the uncle blocks mined by the selfish pool can always be referenced with block distance one, i.e., the maximum uncle reward $7/8 K_s$ using the function $K_u(\cdot)$. In contrast, the honest miners' uncle blocks cannot obtain such high rewards. To illustrate the explicit situation, we provide the distribution of the honest miners' uncle block with different referencing block distances in Table~\ref{tab:distribution} given $\gamma = 0.5$. The results show that with the increase of $\alpha$, the average referencing distance of honest miners' blocks is increasing. Thus, we should decrease the reward for uncle blocks with distance one and increase the reward for the uncle blocks with longer distances. In particular, we can simply set the uncle reward function $K_u(\cdot)$) as a fixed value, say $K_u = 4/8 K_s$, if uncle blocks' referencing block distance is between $1$ and $6$. We recompute the threshold of making selfish mining profitable using this new function and find that when $\gamma = 0.5$, the threshold increases from $0.054$ to $0.163$ in scenario $1$, and from $0.270$ to $0.356$ in scenario $2$. In other words, this simple change makes
it harder for the selfish pool to be profitable.

\begin{table}
\caption{The distribution of honest miners' uncle block with different referencing block distances}
\label{tab:distribution}
\begin{center}
\begin{tabular}{ccc}
\toprule
   Referencing distance & $\alpha = 0.3$ & $\alpha = 0.45$ \\
\midrule
    1   &  0.527& 0.284\\  
    2   &  0.295& 0.249\\  
    3   &  0.111& 0.171\\  
    4   &  0.043& 0.125\\  
    5   &  0.017& 0.096\\  
    6   &  0.007& 0.075\\   
\hline
    Expectation & 1.75 & 2.72 \\
  \bottomrule
\end{tabular}
\end{center}
\end{table}

\section{Related Work} \label{sec:related}
The research of selfish mining is mostly focused on Bitcoin with roughly two directions:
1) optimizing the selfish mining strategies in order to increase the revenue and lower the threshold of launching selfish mining attacks; 2) proposing defense mechanisms. In \cite{eyal2014majority}, Eyal and Sirer developed a Markov process to model the Selfish-Mine Strategy and to evaluate the selfish pool's relative revenue.
Moreover, they proposed a uniform tie-breaking defense against selfish mining, which is adopted in Ethereum.
Inspired by this seminal paper, Sapirshtein et al. \cite{sapirshtein2016optimal} and Nayak et al. \cite{nayak2016stubborn} demonstrated that by adopting the optimized strategies, the threshold of the hashing power to make selfish mining profitable can be reduced to $23.2\%$ even when honest miners adopt the uniform tie-breaking defense. Furthermore, the authors in \cite{gobel2016bitcoin} took the propagation delay into the analysis of selfish mining.

As for defense mechanisms, Heilman proposed a defense mechanism called Freshness Preferred \cite{heilman2014one}, in which by using the latest unforgeable timestamp issued by a trusted party, the threshold can be increased to $32\%$.  Bahack in \cite{bahack2013theoretical} introduced a fork-punishment rule to make selfish mining unprofitable. Especially, each miner in the system can include fork evidence in their block. Once confirmed, the miner can get half of the total rewards of the winning branch. Solat and Potop-Butucaru \cite{solat2016zeroblock} propose a solution called ZeroBlock, which can make selfish miners' block expire and be rejected by all the honest miners without using forgeable timestamps. In \cite{zhang2017publish}, the authors proposed a backward-compatible defense mechanism called weighted FRP which considers the weights of the forked chains instead of their lengths. This is similar in spirit to the GHOST protocol \cite{sompolinsky2015secure}.

There exist very few studies about the selfish mining attack in Ethereum. The work by Gervais et al. \cite{gervais2016security} is among the first to develop a quantitative framework to analyze selfish mining as well as double-spending in various PoW blockchains. 
Particularly, they developed optimal selfish-mining strategies for various PoW blockchains.
However, their work didn't consider the general functions of uncle and nephew rewards. Instead, they focused on a special case when the uncle reward is always $\frac{7}{8}$ of the block reward.
The author in \cite{unclemining} proposed to exploit the flaw of difficulty adjustment to mine additional uncle blocks, which is shown less profitable than our selfish mining strategy. In \cite{ritz2018impact}  Ritz and Zugenmaier built a Monte Carlo simulation platform to quantify the security of the Ethereum after \textsc{EIP100}. However, their paper contains no mathematical analysis and cannot directly capture the effects of uncle block rewards and nephew rewards.  

\section{Conclusion And Future Work} \label{sec:conclude}
In this paper, we have proposed a Markov model to analyze a selfish mining strategy in Ethereum. Our model enables us to evaluate the impact of the uncle and nephew rewards, which is generally missing in the previous analysis for selfish mining in Bitcoin. 
In particular, we have shown how these rewards influence the security of Ethereum mining. Additionally, we have computed the hashing power threshold of making selfish mining profitable under different scenarios, which is essential for us to evaluate the security of Ethereum mining and to design
new reward functions.

As one of our major findings, we notice that it is important to consider uncle blocks when adjusting the mining difficulty level. Otherwise, Ethereum would be much more vulnerable to selfish
mining than Bitcoin. This finding supports the emendation adopted by the third milestone of Ethereum. However, once the mining mechanism is changed, 
the selfish pool is likely to change its new mining strategies in order to maximize its own profit. We leave the design of new mining strategies as our future work. We believe that our analysis developed in this paper (especially the probabilistic tracking) would be useful in studying other mining strategies. 

\section*{Acknowledgement}
We sincerely thank the anonymous reviewers for their valuable comments and feedback
as well as Arthur Gervais for useful discussions on his landmark paper \cite{gervais2016security}. This work was supported by NSERC Discovery Grants RGPIN-2016-05310.

\bibliographystyle{IEEEtran}
\bibliography{bibli}

\appendix
\subsection{The Multiple Summations Function} \label{appen:function}
The function $f(x, y, z)$ used in Sec.~\ref{sec:distribution} involves multiple summations when $z > 1$. We provide several examples to 
explain this function.

\begin{example}[$z = 1$, $x \ge y + 2$]
\begin{align*}
    f(x, y, 1) &= \sum_{s_1 = y + 2}^{x} 1 \\
    &= x - y - 1.
\end{align*}
\end{example}

\begin{example}[$z = 2$, $x \ge y + 2$]
\begin{align*}
    f(x, y, 2) &= \sum_{s_2 = y + 2}^x \sum_{s_1 = y + 1}^{s_2} 1 \\
    &= \sum_{s_2 = y + 2}^x (s_2 - y) \\
    &= 2 + \cdots + (x - y) \\
    &= \frac{(x - y - 1)(x - y + 2)}{2}.
\end{align*}
\end{example}

\subsection{Reward Analysis} \label{appen:reward}
\begin{lemma} \label{lemma:R1}
Consider a new block associated with a state transition from state $(i, j)$ with $i - j \ge 2$.
It will be a regular block with probability $1$ if and only if it is mined by the selfish pool.
\end{lemma} 
\begin{proof}
Suppose that the current system state is $(i, j)$ with $i - j \ge 2$. 
If the selfish pool mines a new block, then the state becomes $(i + 1, j)$.
Now, the private branch has an advantage of at least $3$ blocks 
(since $i + 1 - j \ge 3$) over public branches.
As the system evolves, the private branch will be published with probability $1$ and become
part of the system main chain, according to 
Algorithm~\ref{Algo1}. In other words, the new block will be a regular block with probability $1$.
On the other hand, if some honest miners mine a new block, we consider two cases.
\begin{enumerate}
    \item $i - j = 2$. In this case, we have $L_h = j + 1$ and $L_s = L_h + 1$. Hence, the pool 
    will publish its private branch and the new block becomes a stale block.
    \item $i - j \ge 3$. In this case, the system state becomes either $(i, j + 1)$ or $(i - j, 1)$.
    The private branch has an advantage of at least $3$ blocks (since $i - j \ge 3$). Hence, 
    it will be published with probability $1$. In other words, the new block will be a stale block
    with probability $1$.
\end{enumerate}
\end{proof}

We are now ready to analyze every state transition. We call a new block associated with a transition
a \emph{target} block.

\emph{Case 1: $(0,0) \overset{\beta}{\rightarrow} (0,0)$} 

In this case, the target block generated by some honest miners will be adopted by all the miners. Thus, it will be a regular block and receive a static reward $K_s$. 

\emph{Case 2: $(0,0) \overset{\alpha}{\rightarrow} (1,0)$} 

In this case, the selfish pool produces the target block, keeps it private, and continues mining on it. First, we analyze the static reward by determining whether the target block will be a regular block or not. To this end, we consider the following two subcases, which are illustrated in Fig.~\ref{fig:incen1}.

\begin{figure}[ht]
\centering
\includegraphics[width=0.8\linewidth]{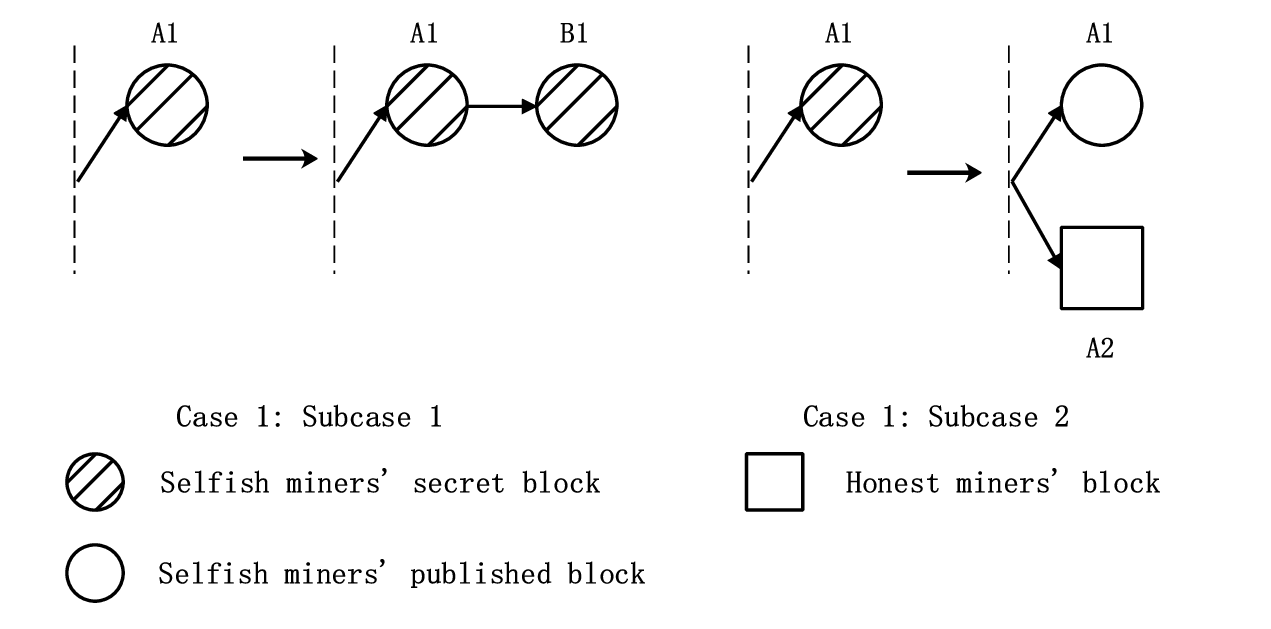}
\caption{The two subcases of Case 1: 1) the selfish pool mines a subsequent block; 2) some honest miners mine a new block.}
\vspace{-2mm}
\label{fig:incen1}
\end{figure}

\begin{enumerate}
    \item \emph{Subcase 1}: The subsequent block is mined by the pool, 
    which happens with probability $\alpha$. As a result, the pool owns a lead of two blocks. By Lemma \ref{lemma:R1}, the target block will be a regular block and receive a static reward of $K_s$. 
    \item \emph{Subcase 2}: The subsequent block is mined by some honest miner, which happens with probability $\beta$. Then, the pool will publish this target block. To determine whether it will be a regular block, we need to consider the following three subsubcases. See Fig. \ref{fig:incen2} for an illustration. 
    
    \emph{Subsubcase 1}: The pool mines a new block on its private branch and publishes it immediately. 
    (This happens with probability $\alpha$.) Now, the target block becomes a regular block. 
    
    \begin{figure}[ht]
    \centering
    \includegraphics[width=0.8\linewidth]{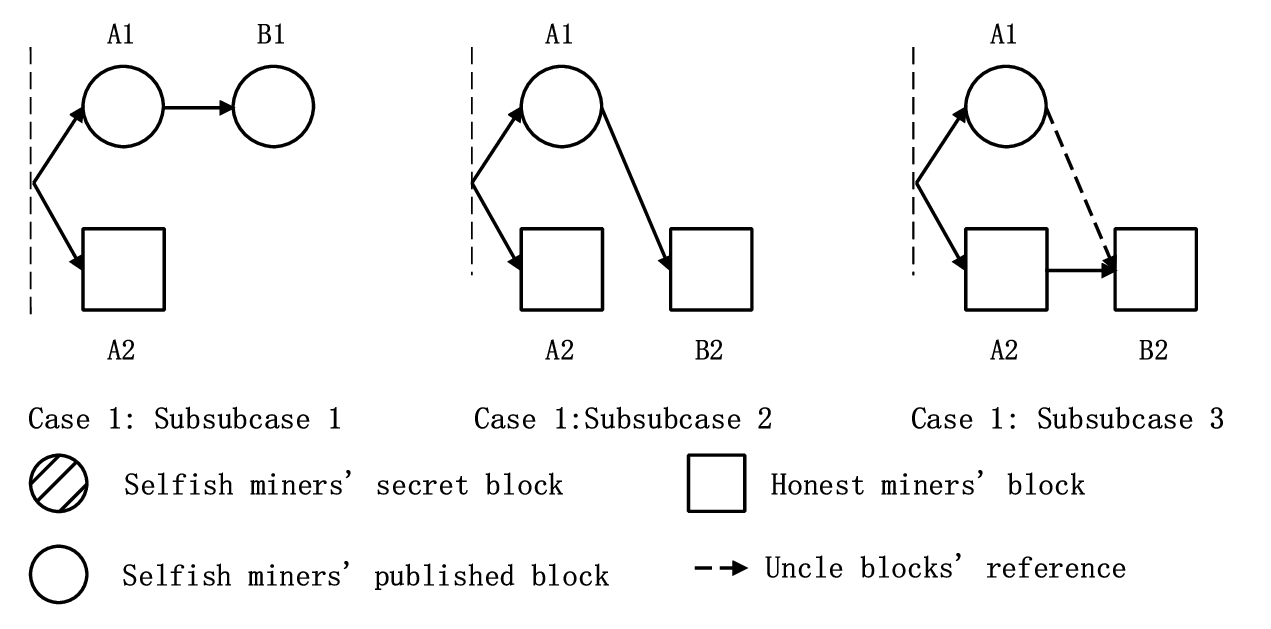}
    \caption{The three subsubcases of Case 1: 1) the selfish pool mines a new block and publishes it; 2) the honest miners find a new block on the target block; 3) some honest miners mine a new block and reference the target block.}
    \vspace{-2mm}
    \label{fig:incen2}
    \end{figure}
    
    \emph{Subsubcase 2}: Some honest miners mine a new block on the target block. (This happens with probability $\beta \gamma$.) Now, the target block becomes a regular block. 
    
    \emph{Subsubcase 3}: Some honest miners mine a new block and reference the target block. (This happens with probability $\beta (1 - \gamma)$.) Now, the target block becomes an uncle block. 
\end{enumerate}
To sum up, the target block in \emph{Case 2} will eventually be a regular block with probability $\alpha + \alpha \beta + \beta^2 \gamma$ and be an uncle block with probability $\beta^2 (1 - \gamma)$. 
(Note that these two probabilities sum up to $1$.)

Next, we analyze the uncle and nephew rewards associated with the target block. Based on our previous case-by-case discussion, only in subsubcase $3$, the target block will be an uncle block. 
Also, the distance between the target block and its nephew block is $1$. 
As such, the target block will bring the pool an uncle reward of $K_u(1)$ and some honest miners will receive a nephew reward of $K_n$. (This happens with probability $\beta^2 (1 - \gamma)$.)

\emph{Case 3: $(1,0) \overset{\alpha}{\rightarrow} (2,0)$} 

In this case, the pool produces the target block, keeps it private, and continues mining on it.
By Lemma \ref{lemma:R1}, the target block will be a regular block and receive a static reward
of $K_s$.

\emph{Case 4: $(1,0) \overset{\beta}{\rightarrow} (1,1)$}

In this case, some honest miners mine the target block, then the pool publishes its private block. 
First, we analyze the static reward by determining whether the target block is a regular block. To this end, we consider the following subcases. See Fig.~\ref{fig:incen3} for an illustration. 
\begin{enumerate}
    \item \emph{Subcase 1}: The pool mines a new block on its private branch, references the target block, and publishes its private branch. (This happens with probability $\alpha$.) Now, the target block becomes an uncle block. 

    \item \emph{Subcase 2}: Some honest miners mine a new block not on the target block and reference
    the target block. 
    (This happens with probability $\beta \gamma$.) Now, the target block becomes an uncle block.
    
    \item \emph{Subcase 3}: Some honest miners mine a new block on the target block. 
    (This happens with probability $\beta (1 - \gamma)$.) Now, the target block becomes a regular block. 
    
    \begin{figure}[ht]
    \centering
    \includegraphics[width=0.8\linewidth]{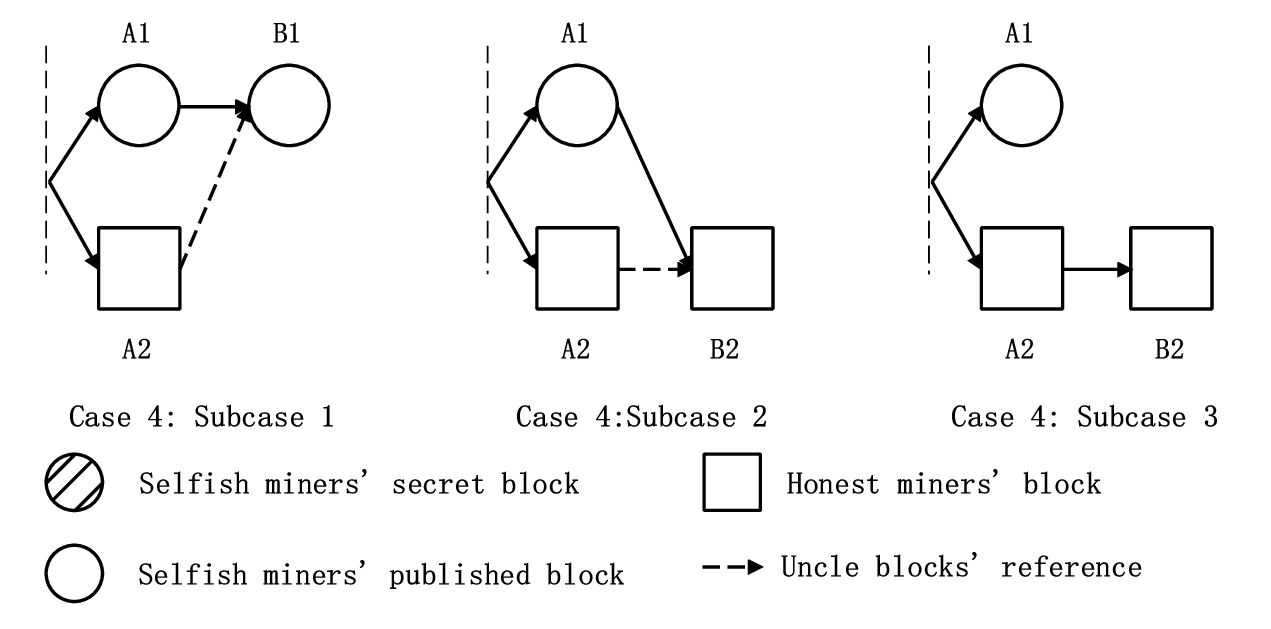}
    \caption{The three subcases of Case 4: 1) the selfish pool mines a new block on its branch and references the target block; 2) some honest miners find a new block not on the target block and reference the target block; 3) some honest miners mine a new block on the target block.}
    \vspace{-2mm}
    \label{fig:incen3}
    \end{figure}
\end{enumerate}
To sum up, the target block will eventually be a regular block with probability $\beta ( 1 - \gamma)$, and be an uncle block with probability $\alpha + \beta \gamma$. 

Next, we analyze the uncle and nephew rewards associated with the target block.  Based on our previous case-by-case discussion, the target block will become an uncle block only in subcases $1$ and $2$,
where the distance is $1$. Thus, the target block will bring honest miners an uncle reward of $K_u(1)$.
As for the nephew reward, in \emph{Subcase 1}, the pool receives it, and in \emph{Subcase 2}, some honest miners receive it. To sum up, honest miners will receive an uncle block reward of $K_u(1)$ with probability $\alpha + \beta \gamma$, receive a nephew reward of $K_n$ with probability $\beta \gamma$, and the pool will receive a nephew reward of $K_n$ with probability $\alpha$.

\emph{Case 5: $(1,1) \overset{1}{\rightarrow} (0,0)$} 

In this case, the target block will always be a regular block no matter who mines it.
Hence, the pool receives a static reward with probability $\alpha$, and some honest
miners receive a static reward with probability $\beta$.

\emph{Case 6: $(i,j) \overset{\alpha}{\rightarrow} (i+1,j)$ with $i \geq 2$ and $j \geq 0$} 

In this case, the pool mines the target block, keeps it private, and continues mining. By Lemma \ref{lemma:R1}, the target block will eventually become a regular block, receiving a static reward of $K_s$. 

\emph{Case 7: $(i,j) \overset{\beta \gamma}{\rightarrow} (i-j,1)$  with $i - j \geq 3$ and $j \geq 1$} 

In this case, some honest miners mine the target block on a public branch that is a prefix
of the private branch. Then, the pool publishes its first unpublished block. 
By Lemma \ref{lemma:R1}, the target block will eventually become an uncle block. 

Next, we analyze the uncle and nephew rewards associated with the target block. 
We begin with the special case of $(4, 1) \rightarrow (3, 1)$ before discussing the general case.
We consider the following three subcases. 

\begin{enumerate}
    \item \emph{Subcase 1}: The pool mines a subsequent block and references the target block.
    See Fig.~\ref{fig:incen4} for an illustration. (This happens with probability $\alpha$.)
    By Lemma~\ref{lemma:R1}, the target block will become an uncle block, receiving an uncle reward of $K_u(3)$. The pool will receive a nephew reward.
    
    \begin{figure}[ht]
    \centering
    \includegraphics[width=0.85\linewidth]{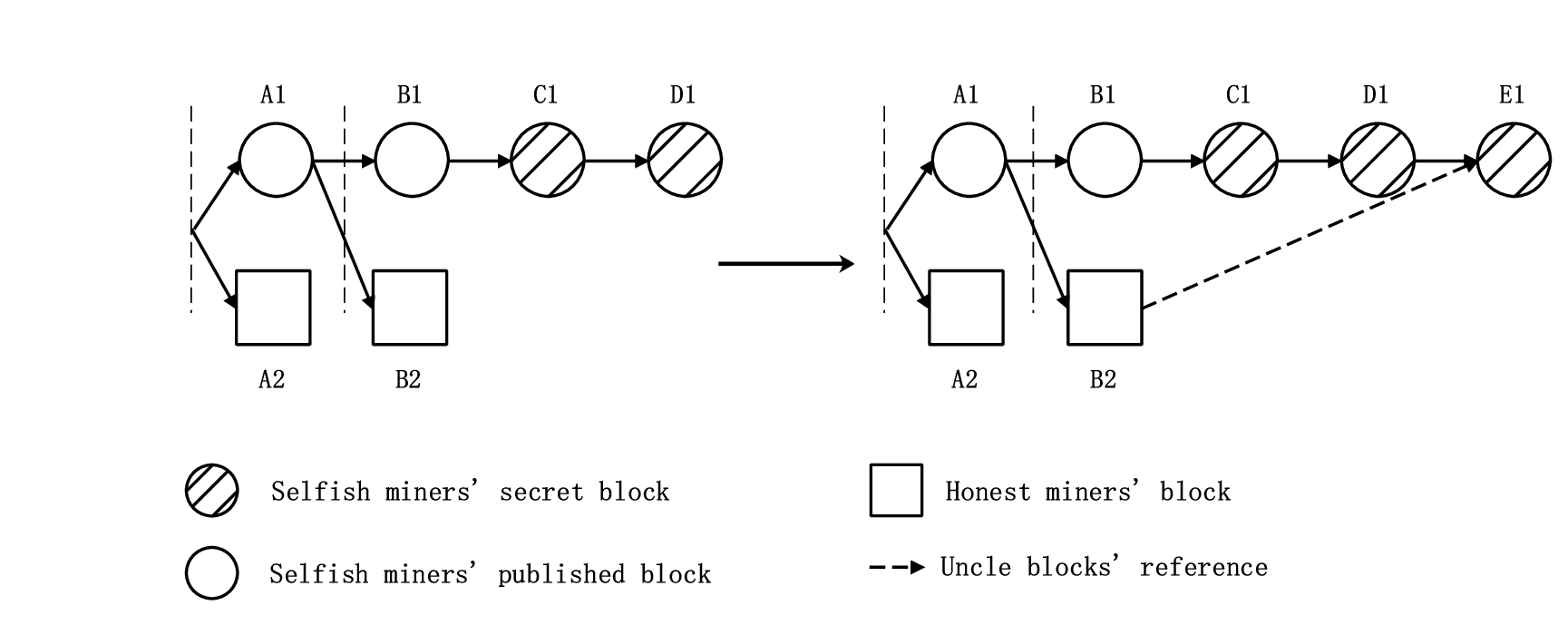}
    \caption{The subcase 1 of Case 7, in which the pool mines a subsequent block, references the target block and eventually wins the associated nephew reward.}
    \vspace{-2mm}
    \label{fig:incen4}
    \end{figure}
    
    \item \emph{Subcase 2}: Some honest miners mine a subsequent block on the target block.
    (This happens with probability $\beta (1 - \gamma)$.) Then, the pool will publish its private branch.
    See Fig. \ref{fig:incen5} for an illustration. To determine the uncle and nephew rewards,
    we need to consider the following subsubcases. 
    
    \begin{figure}[ht]
    \centering
    \includegraphics[width=0.8\linewidth]{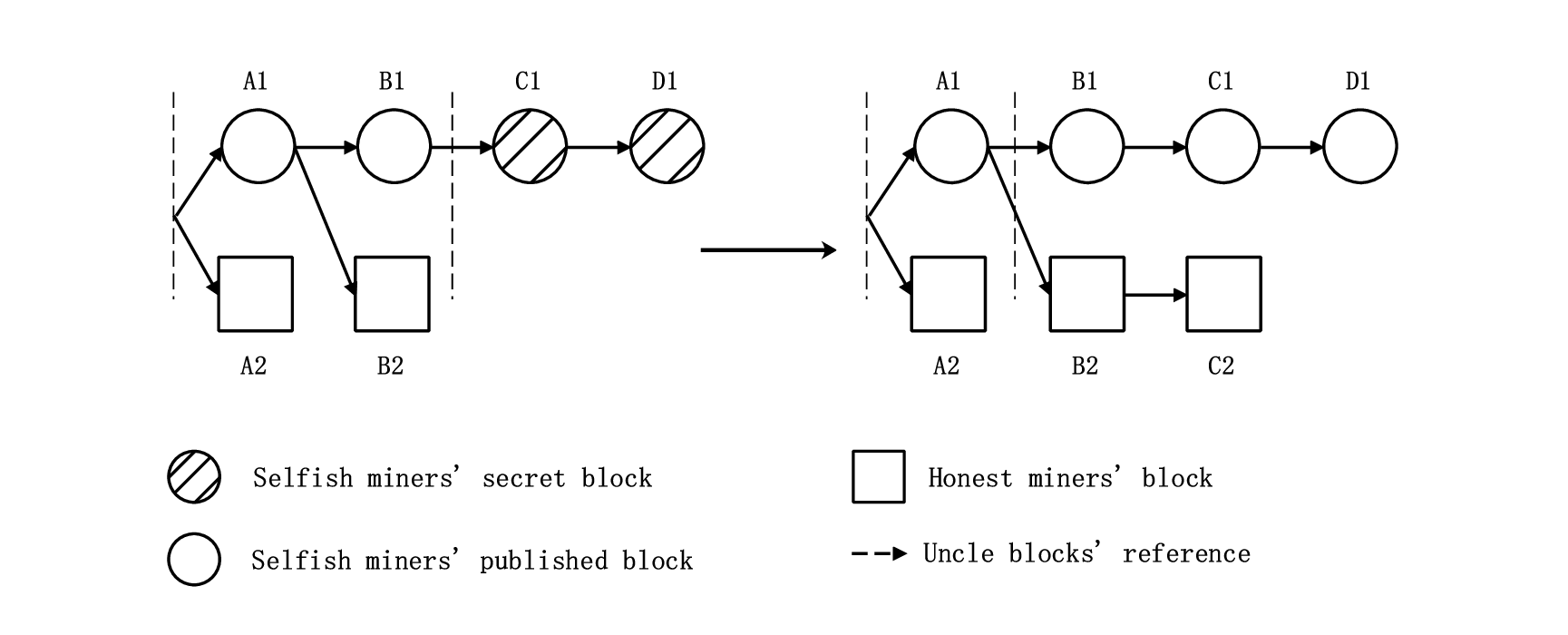}
    \caption{The subcase 2 of Case 7, in which some honest miners mine a subsequent block on the target block.}
    \vspace{-2mm}
    \label{fig:incen5}
    \end{figure}
    
    \emph{Subsubcase 1}: Some honest miners mine a new block and reference the target block.
    See Fig.~\ref{fig:incen6}. This subsubcase happens with probability $\beta (1 - \gamma) \beta$.
    The honest miner receives a nephew reward, and the target block receives an uncle reward of
    $K_u(3)$ since the distance is $3$.
    
    \begin{figure}[ht]
    \centering
    \includegraphics[width=0.85\linewidth]{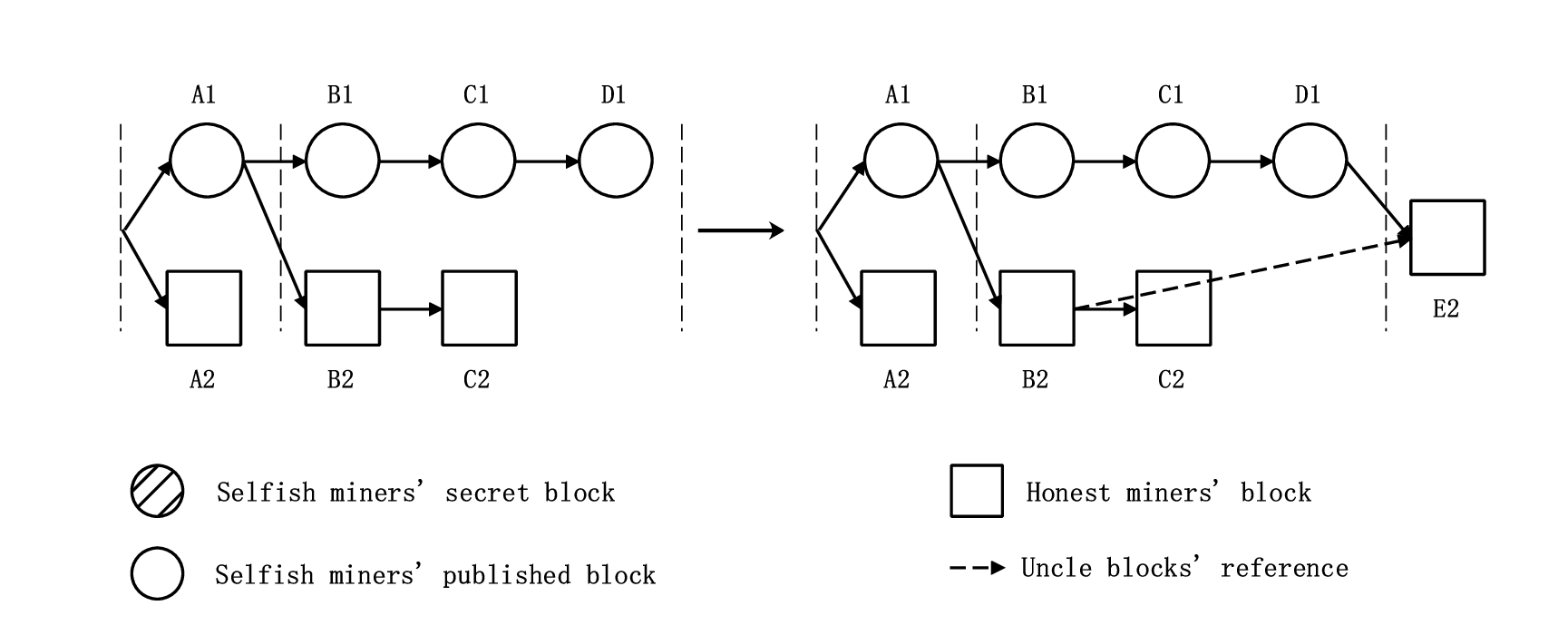}
    \caption{The subsubcase 1 of Case 7, in which some honest miners mine a block in state $(0,0)$ and win the associated nephew reward.}
    \vspace{-2mm}
    \label{fig:incen6}
    \end{figure}
    
    \emph{Subsubcase 2}: The pool mines a new block and keeps it private. 
    This subsubcase happens with probability $\beta (1 - \gamma) \alpha$.
    Now, if the new block later becomes a regular block (with probability 
    $\alpha + \alpha \beta + \beta^2 \gamma$ due to the discussion for  
    \emph{Case 2}), the pool will receive 
    a nephew reward and the target block will receive an uncle reward of $K_u(3)$.
    Otherwise, if the new block later becomes a stale block 
    (with probability $\beta^2 (1 - \gamma)$ due to the discussion for  
    \emph{Case 2}), some honest miners will receive 
    a nephew reward and the target block will again receive an uncle reward of $K_u(3)$.
    
    \begin{figure}[ht]
    \centering
    \includegraphics[width=0.8\linewidth]{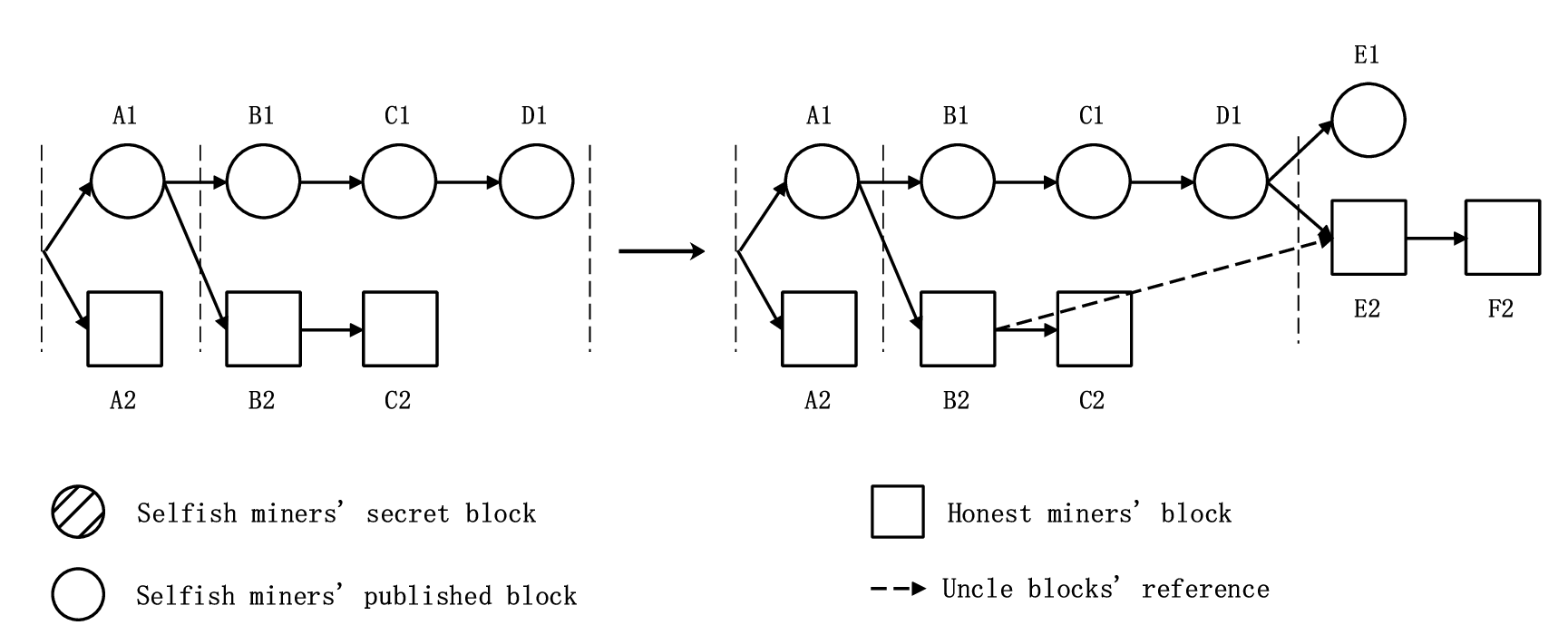}
    \caption{The subsubcase 2 of Case 7, in which the selfish pool mines a new block in state $(0,0)$, but finally loses to the nephew reward.}
    \vspace{-2mm}
    \label{fig:incen7}
    \end{figure}
    
    \item \emph{Subcase 3}: Some honest miners mine a subsequent block, not on the target block.
    (This happens with probability $\beta \gamma$.)
    In terms of the analysis of uncle and nephew rewards, \emph{Subcase 3} is the same as \emph{Subcase 2}.
\end{enumerate}

To sum up, for the special case of $(4, 1) \rightarrow (3, 1)$, the target block will always
receive an uncle reward of $K_u(3)$. As for the nephew reward, one can draw a tree diagram
summarizing all the subcases discussed above and conclude that it will be received by honest miners
with probability $\beta^2 (1 + \alpha \beta (1 - \gamma))$ and by the pool
with probability $1 - \beta^2 (1 + \alpha \beta (1 - \gamma))$.

Here, we note that the probability $\beta^2 (1 + \alpha \beta (1 - \gamma))$ (that the nephew reward will be received by honest miners) can be explained as follows. 
First, the honest miners have to ``push'' the system state from $(3, 1)$ to $(0, 0)$
while the pool mines nothing. (Otherwise, the pool will receive the nephew reward by 
Lemma~\ref{lemma:R1}.) This happens with probability $\beta$. Second, starting from
$(0, 0)$, the honest miners can win the nephew reward with probability $\beta (1 + \alpha \beta (1 - \gamma))$. This interpretation allows us to analyze the general case.

First, the honest miners have to ``push'' the system state from $(i-j, 1)$ to $(0, 0)$
while the pool mines nothing. This happens with probability $\beta^{i-j-2}$.
Second, starting from
$(0, 0)$, the honest miners can win the nephew reward with probability $\beta (1 + \alpha \beta (1 - \gamma))$. Therefore, the nephew reward will be received by honest miners with 
probability $\beta^{i-j-1} (1 + \alpha \beta (1 - \gamma))$ and by the pool 
with probability $1 - \beta^{i-j-1} (1 + \alpha \beta (1 - \gamma))$.

\emph{Case 8: $(i,j) \overset{\beta \gamma}{\rightarrow} (0,0)$ with $i - j=2$ and $j \geq 1$} 

In this case, some honest miners mine the target block. Then, the pool
publishes its private branch. Now, the target block becomes an uncle block. Similar to 
\emph{Case 7}, the target block will receive an uncle reward of $R_u(2)$, and the nephew 
reward will be received by honest miners with probability $\beta (1 + \alpha \beta (1 - \gamma))$
and by the pool with probability $1 - \beta (1 + \alpha \beta (1 - \gamma))$.

\emph{Case 9: $(2,0) \overset{\beta}{\rightarrow} (0,0)$ with $i \geq 2$} 

In this case, some honest miners mine the target block. Then, the pool
publishes its private branch. The remaining discussion is the same as \emph{Case 8}.

\emph{Case 10: $(i,0) \overset{\beta}{\rightarrow} (i,1)$ with $i \geq 3$} 

In this case, some honest miners mine the target block. Then, the pool
publishes its first unpublished block. By Lemma~\ref{lemma:R1},
the target block will eventually become an uncle block.
Similar to the discussion for \emph{Case  7},
we conclude that the target block will receive an uncle reward of $K_u(i)$,
and the nephew reward will be received by honest miners with 
probability of $\beta^{i-1} (1 + \alpha \beta (1 - \gamma))$
and by the pool with probability $1 - \beta^{i-1} (1 + \alpha \beta (1 - \gamma))$.

\emph{Case 11: $(i,j) \overset{\beta (1-\gamma)}{\rightarrow} (i,j + 1)$ with $i-j \geq 3$ and $j \geq 1$} 

In this case, some honest miners mine the target block. Then, the pool
publishes its first unpublished block. By Lemma~\ref{lemma:R1}, the target block
will eventually become a stale block. Since its parent block is not in the system main chain,
the target block will not be an uncle block.

\emph{Case 12: $(i,j) \overset{\beta (1-\gamma)}{\rightarrow} (0,0)$ with $i -j=2$ and $j \geq 1$} 

In this case, some honest miners mine the target block. Then, the pool
publishes its private branch. Similar to \emph{Case 11}, the target block will neither be a regular block nor an uncle block.

\end{document}